\documentclass[journal]{IEEEtran}
\usepackage{cite}

\ifCLASSINFOpdf
   \usepackage[pdftex]{graphicx}
   \graphicspath{{./fig/}{./bio/}}
   \DeclareGraphicsExtensions{.pdf,.png}
\else
\fi
\usepackage{amsmath}
\usepackage{amsfonts}
\usepackage{amssymb}
\usepackage{amsthm}
\usepackage{mathtools}
\usepackage{algorithm}
\usepackage{algpseudocode}
\usepackage{xcolor}
\usepackage{enumerate}
\usepackage{multirow}
\DeclareMathOperator*{\argmin}{arg\,min}
\DeclareMathOperator*{\minimize}{minimize}
\DeclareMathOperator*{\st}{subject\,to}
\DeclareMathOperator{\Diag}{Diag}
\DeclareMathOperator{\E}{E}
\DeclareMathOperator{\Imag}{\Im}
\DeclareMathOperator{\Master}{\mathcal M}
\DeclareMathOperator{\Prox}{Prox}
\DeclareMathOperator{\Real}{\Re}
\DeclareMathOperator{\Slave}{\mathcal S}
\DeclareMathOperator{\Supp}{Supp}
\newcommand{\conj}[1]{%
\overline{#1}%
}
\newcommand{\isdef}{\coloneqq}
\newcommand{\had}{\circ}
\newcommand{\cplx}{\mathbb C}
\newcommand{\intg}{\mathbb Z}
\newcommand{\real}{\mathbb R}
\newcommand{\bfa}{\mathbf a}
\newcommand{\bfb}{\mathbf b}

\newcommand{\bfd}{\mathbf d}
\newcommand{\bfe}{\mathbf e}
\newcommand{\bfg}{\mathbf g}
\newcommand{\bfh}{\mathbf h}
\newcommand{\bfp}{\mathbf p}
\newcommand{\bfr}{\mathbf r}
\newcommand{\bfs}{\mathbf s}
\newcommand{\bfu}{\mathbf u}
\newcommand{\bfv}{\mathbf v}
\newcommand{\bfw}{\mathbf w}
\newcommand{\bfx}{\mathbf x}
\newcommand{\bfy}{\mathbf y}
\newcommand{\bfz}{\mathbf z}
\newcommand{\bfzero}{\mathbf 0}
\newcommand{\bfA}{\mathbf A}
\newcommand{\bfB}{\mathbf B}
\newcommand{\bfC}{\mathbf C}
\newcommand{\bfD}{\mathbf D}
\newcommand{\bfE}{\mathbf E}
\newcommand{\bfF}{\mathbf F}
\newcommand{\bfG}{\mathbf G}
\newcommand{\bfI}{\mathbf I}
\newcommand{\bfP}{\mathbf P}
\newcommand{\bfR}{\mathbf R}
\newcommand{\bfS}{\mathbf S}
\newcommand{\bfV}{\mathbf V}
\newcommand{\bfX}{\mathbf X}
\newcommand{\bfY}{\mathbf Y}
\newcommand{\bfZ}{\mathbf Z}
\newcommand{\bfgamma}{\boldsymbol\gamma}
\newcommand{\bfdelta}{\boldsymbol\delta}
\newcommand{\bfepsilon}{\boldsymbol\epsilon}
\newcommand{\bftheta}{\boldsymbol\theta}
\newcommand{\bfphi}{\boldsymbol\phi}
\newcommand{\bfOmega}{\boldsymbol\Omega}
\newcommand{\Let}[2]{$#1 \leftarrow #2$}

\theoremstyle{definition}
\newtheorem{corollary}{Corollary}
\newtheorem{definition}{Definition}
\newtheorem{lemma}[corollary]{Lemma}
\newtheorem{proposition}[corollary]{Proposition}
\newtheorem*{remark}{Remark}
\newtheorem*{notation}{Notation}
\newtheorem{theorem}[corollary]{Theorem}

\begin{document}
%
\title{Single-Look Multi-Master SAR Tomography: An Introduction}
%
%
%

\author{Nan~Ge,
        Richard~Bamler,~\IEEEmembership{Fellow,~IEEE,}
        Danfeng~Hong,~\IEEEmembership{Member,~IEEE,}
        and~Xiao~Xiang~Zhu,~\IEEEmembership{Senior~Member,~IEEE}
\thanks{This work is jointly supported by the European Research Council (ERC) under the European Union's Horizon 2020 research and innovation programme (grant agreement No. [ERC-2016-StG-714087], Acronym: \textit{So2Sat}), and by the Helmholtz Association
through the Framework of Helmholtz Artificial Intelligence Cooperation Unit (HAICU) - Local Unit ``Munich Unit @Aeronautics, Space and Transport (MASTr)'' and Helmholtz Excellent Professorship ``Data Science in Earth Observation - Big Data Fusion for Urban Research''.
\emph{Corresponding author: Xiao~Xiang~Zhu.}}

\thanks{N.~Ge and D.~Hong are with the Remote Sensing Technology Institute (IMF), German Aerospace Center (DLR), 82234 Wessling, Germany (e-mail: Nan.Ge@dlr.de, Danfeng.Hong@dlr.de).}
\thanks{R.~Bamler is with the Remote Sensing Technology Institute (IMF), German Aerospace Center (DLR), 82234 Wessling, Germany, and with the Chair of Remote Sensing Technology, Technical University of Munich (TUM), 80333 Munich, Germany (e-mail: richard.bamler@dlr.de).}
\thanks{X.~X.~Zhu is with the Remote Sensing Technology Institute (IMF), German Aerospace Center (DLR), 82234 Wessling, Germany, and with Signal Processing in Earth Observation (SiPEO), Technical University of Munich (TUM), 80333 Munich, Germany (e-mail: xiaoxiang.zhu@dlr.de).}
}

%
%

\markboth{}%
{Shell \MakeLowercase{\textit{et al.}}: Bare Demo of IEEEtran.cls for IEEE Journals}
%



\maketitle

\begin{abstract}
This paper addresses the general problem of single-look multi-master SAR tomography. For this purpose, we establish the single-look multi-master data model, analyze its implications for single and double scatterers, and propose a generic inversion framework. The core of this framework is nonconvex sparse recovery, for which we develop two algorithms: one extends the conventional nonlinear least squares (NLS) to the single-look multi-master data model, and the other is based on bi-convex relaxation and alternating minimization (BiCRAM). We provide two theorems for the objective function of the NLS subproblem, which lead to its analytic solution up to a constant phase angle in the one-dimensional case. We also report our findings from the experiments on different acceleration techniques for BiCRAM. The proposed algorithms are applied to a real TerraSAR-X data set, and validated with height ground truth made available via a SAR imaging geodesy and simulation framework. This shows empirically that the \emph{single-master} approach, if applied to a single-look \emph{multi-master} stack, can be insufficient for layover separation, and the \emph{multi-master} approach can indeed perform slightly better (despite being computationally more expensive) even in the case of single scatterers. Besides, this paper also sheds light on the special case of single-look bistatic SAR tomography, which is relevant for current and future SAR missions such as TanDEM-X and Tandem-L.
\end{abstract}

\begin{IEEEkeywords}
Synthetic aperture radar (SAR), bistatic SAR, TanDEM-X, Tandem-L, SAR tomography, sparse recovery, nonconvex optimization.
\end{IEEEkeywords}

%
\IEEEpeerreviewmaketitle

\section{Introduction} \label{sec_introduction}
%
%
%
%



\IEEEPARstart{S}{ynthetic} aperture radar (SAR) tomography is an interferometric SAR (InSAR) technique that reconstructs a three-dimensional far field from two-dimensional (2-D) azimuth-range measurements of radar echoes \cite{reigber2000first, gini2002layover, fornaro2003three}. In the common case of spaceborne repeat-pass acquisitions, scatterers' motion can also be modeled and estimated \cite{lombardini2005differential, fornaro2008four, zhu2011let}. SAR tomography is sometimes considered as an extension of persistent scatterer interferometry (PSI) \cite{ferretti2001permanent, colesanti2003sar, adam2005development} to the multi-scatterer case, although the inversion of the latter is performed on double-difference phase observations of persistent scatterers (PS) \cite{kampes2006radar}. Extensive efforts were devoted to improving the super-resolution power, robustness and computational efficiency of tomographic inversion in urban scenarios (e.g., \cite{budillon2010three, zhu2010tomographic, aguilera2012multisignal, schmitt2012compressive, fornaro2014caesar, zhu2015joint, ge2018spaceborne, shi2018fast, shi20183d, shi2019nonlocal}).

\begin{table*}[!tbp]
    \renewcommand{\arraystretch}{1.3}
    \caption{A classification of tomographic SAR algorithms}
    \label{tab_alg}
    \centering
    \begin{tabular}{l|p{4.5cm}|p{4.5cm}|}
        \multicolumn{1}{l}{} & \multicolumn{1}{c}{Single-Master} & \multicolumn{1}{c}{Multi-Master} \\
        \cline{2-3}
        \parbox[t]{2mm}{\multirow{5}{*}{\rotatebox[origin=c]{90}{Single-Look}}}
        & Reigber \& Moreira (2000) & Zhu \& Bamler (2012)$^\dagger$ \\
        & Fornaro et al.\ (2003, 2005, 2008) & Ge \& Zhu (2019)$^\dagger$ \\
        & Budillon et al.\ (2010) & \\
        & Zhu \& Bamler (2010a, 2010b, 2011) & \\
        & Etc.\ & \\
        \cline{2-3}
        \parbox[t]{2mm}{\multirow{5}{*}{\rotatebox[origin=c]{90}{Multi-Look}}}
        & Aguilera et al.\ (2012) & Gini et al.\ (2002) \\
        & Schmitt \& Stilla (2012) & Lombardini (2005) \\
        & Liang et al.\ (2018) & Duque et al.\ (2009, 2010, 2014) \\
        & Shi et al.\ (2019) & Fornaro et al.\ (2014) \\
        & Etc.\ & Etc.\ \\
        \cline{2-3}
        \multicolumn{1}{l}{} & \multicolumn{2}{l}{$^\dagger$ Incorrectly uses the single-look single-master data model}
    \end{tabular}
\end{table*}

The publications on SAR tomography can be roughly classified into the following four categories (see also Tab.~\ref{tab_alg}). Note that those listed below were only hand-picked, and we have no intention to provide a complete list.
\begin{itemize}
    \item \emph{Single-look single-master:} \\
    Reigber \& Moreira (2000) did the pioneering work on airborne SAR tomography by densifying sampling via the integer interferogram combination technique and subsequently employing discrete Fourier transform on an interpolated linear array of baselines \cite{reigber2000first}. Fornaro et al.\ (2003, 2005, 2008) paved the way for spaceborne SAR tomography with long-term repeat-pass acquisitions and proposed to use more advanced inversion techniques such as truncated singular value decomposition \cite{fornaro2003three, fornaro2005three, fornaro2008four}. Zhu \& Bamler (2010a) provided the first demonstration of SAR tomography with very high resolution spaceborne SAR data by using Tikhonov regularization and nonlinear least squares (NLS) \cite{zhu2010very}. Budillon et al.\ (2010) and Zhu \& Bamler (2010b) introduced compressive sensing techniques to tomographic inversion under the assumption of a compressible far-field profile. Zhu \& Bamler (2011) proposed a generic algorithm (named SL1MMER) that is composed of spectral estimation, model-order selection and debiasing \cite{zhu2011super}.
    \item \emph{Single-look multi-master\footnote{In this context, ``multi-master'' can be interpreted as ``not single-master'' (see also our definition in Sec.~\ref{sec_multimaster_data}).}:} \\
    To the best of our knowledge, the publications in this category are rather scarce. Zhu \& Bamler (2012) extended the Tikhonov regularization, NLS and compressive sensing approaches to a mixed TerraSAR-X and TanDEM-X stack by using pre-estimated covariance matrix \cite{zhu2012sparse}. Ge \& Zhu (2019) proposed a framework for SAR tomography using only bistatic or pursuit monostatic acquisitions: non-differential SAR tomography for height estimation by using bistatic or pursuit monostatic interferograms, and differential SAR tomography for deformation estimation by using conventional repeat-pass interferograms and the previous height estimates as deterministic prior \cite{ge2019bistatic}. However, the single-look single-master data model still underlies the algorithms in both publications.
    \item \emph{Multi-look single-master:} \\
    Aguilera et al.\ (2012) exploited the common sparsity pattern among multiple polarimetric channels via distributed compressive sensing \cite{aguilera2012multisignal}. Schmitt \& Stilla (2012) also employed distributed compressive sensing to jointly reconstruct an adaptively chosen neighborhood \cite{schmitt2012compressive}. Liang et al.\ (2018) proposed an algorithm for 2-D range-elevation focusing on azimuth lines via compressive sensing \cite{liang2018urban}. Shi et al.\ (2019) performed nonlocal InSAR filtering before tomographic reconstruction \cite{shi2019nonlocal}.
    \item \emph{Multi-look multi-master:} \\
    In general, any algorithm estimating the auto-correlation matrix belongs to this category. Note that this is closely related to modern adaptive multi-looking techniques that also exploit all possible interferometric combinations \cite{perissin2011repeat, ferretti2011new, parizzi2010adaptive, ansari2017sequential, ansari2018efficient}. Gini et al.\ (2002) investigated the performance of different spectral estimators including Capon, multiple signal classification (MUSIC) and the multi-look relaxation (M-RELAX) algorithm \cite{gini2002layover}. Lombardini (2005) extended SAR tomography to the differential case by formulating it as a multi-dimensional spectral estimation problem and tackled it with higher-order Capon \cite{lombardini2005differential}. Duque et al.\ (2009, 2010) were the first to investigate bistatic SAR tomography by using ground-based receivers and spectral estimators such as Capon and MUSIC \cite{duque2009experimental, duque2010bistatic}. Duque et al.\ (2014) demonstrated the feasibility of SAR tomography using a single pass of alternating bistatic acquisitions, in which the eigendecomposed empirical covariance matrix was exploited for the hypothesis test on the number of scatterers \cite{duque2014single}. Fornaro et al.\ (2014) proposed an algorithm (named CAESAR) employing principal component analysis of the eigendecomposed empirical covariance matrix in an adaptively chosen neighborhood \cite{fornaro2014caesar}.
\end{itemize}
This list has a clear focus on urban scenarios. Needless to say, SAR tomography in forested scenarios involving random volume scattering in canopy and double-bounce scattering between ground and trunk (e.g., \cite{cloude2006polarization, dezan2009tandem, tebaldini2009algebraic, tebaldini2010single, huang2011under, lee2015tandem, pardini2017estimation, cazcarra2019comparison}) also falls in the \emph{multi-look multi-master} category.

Let us follow the common conventions and denote the azimuth, range and elevation axes as $x$, $r$ and $s$, respectively, where $s$ is perpendicular to the $x$-$r$ plane. For the sake of argument, suppose for any sample at the $x$ and $r$ positions, the $N$ single-look complex (SLC) SAR measurements are noiseless. After deramping, each phase-calibrated SLC measurement can be modeled as the Fourier transform $\Gamma$ of the elevation-dependent far-field reflectivity function $\gamma:\real\to\cplx$ at the corresponding wavenumber $k$ \cite{fornaro2005three}:
\begin{equation}
    g_n = \Gamma(k_n) \isdef \int \gamma(s) \exp(-jk_ns) \mathrm ds, \quad n = 1, \ldots, N,
\end{equation}
where $k_n \isdef -4\pi b_n / (\lambda r_0)$ is the $n$th wavenumber determined by the sensor position $b_n$ along an axis $b \parallel s$ w.r.t.\ an arbitrary reference, the radar wavelength $\lambda$, and the slant-range distance $r_0$ w.r.t.\ a ground reference point. Here we consider the non-differential case. An extension to the differential case, in which scatterers' motion is modeled as linear combination of basis functions, is straightforward.

In the \emph{single-look single-master} case, one SLC (say the $i$th, $i \in [N]$), typically near the center of joint orbital and temporal distribution, is selected as the unique master for generating interferograms, i.e., $g_n \conj{g_i} / |g_i|$, $\forall n \in [N] \setminus \{i\}$. This process, which can also be interpreted as a phase calibration step, converts $k_n$ into the wavenumber baseline $\Delta k_n \isdef k_n - k_i$, $\forall n \in [N]$. As a result, the zero position of wavenumber baseline is fixed, i.e., $\Delta k_i=0$. The rationale behind this is, e.g., to facilitate 2-D phase unwrapping for atmospheric phase screen (APS) compensation by smoothing out interferometric phase in $x$-$r$.

Likewise, the data model of random volume scattering is straightforward in the \emph{multi-look multi-master} setting. Suppose $\gamma(s)$ is a white random signal. For any master and slave sampled at $k$ and $k+\Delta k$, respectively, the Van Cittert--Zernike theorem implies that the expectation (due to multi-looking) of the interferogram, being the autocorrelation function $R_{\Gamma\Gamma}$ of $\Gamma$, is the Fourier transform of the elevation-dependent backscatter coefficient function $\sigma_0:\real\to\real$ at $\Delta k$:
\begin{equation} \label{eq_ml_mm}
    \E[\Gamma(k+\Delta k) \conj{\Gamma(k)}] = R_{\Gamma\Gamma}(\Delta k) = \int \sigma_0(s) \exp(-j\Delta ks) \mathrm ds,
\end{equation}
where the property of $\gamma(s)$ being white, i.e., $\E[\gamma(s) \conj{\gamma(s')}] = \sigma_0(s) \delta(s-s')$, is utilized. This leads to an inverse problem similar to the one in the \emph{single-look single-master} case.

This paper, on the other hand, addresses the general problem of SAR tomography using a \emph{single-look multi-master} stack. Such a stack arises when, e.g.,
\begin{itemize}
    \item a stack of bistatic interferograms is used in order to diminish APS, to minimize temporal decorrelation of non-PSs, and to eliminate motion-induced phase for single scatterers \cite{krieger2007tandem, duque2010bistatic, goel2013fusion},
    \item repeat-pass interferograms of small (temporal) baselines are employed so as to limit the corresponding decorrelation effects of non-PSs \cite{berardino2002new, lanari2004small, hooper2008multi}.
\end{itemize}

While both previously mentioned categories have been intensively studied, it is not the case for \emph{single-look multi-master} SAR tomography. To the best of our knowledge, all the existing work to date toward single-look \emph{multi-master} SAR tomography is still incorrectly based on the single-look \emph{single-master} data model \cite{zhu2012sparse, ge2019bistatic}. As will be demonstrated later with a real SAR data set, this approach can be insufficient for layover separation, even if the elevation distance between two scatterers is significantly larger than the Rayleigh resolution. This motivates us to fill the gap in the literature by revisiting the single-look data model in a multi-master multi-scatterer configuration, and by developing efficient methods for tomographic reconstruction. Naturally, this study is also inspired by prospective SAR missions such as Tandem-L that will deliver high-resolution wide-swath bistatic acquisitions in L-band as operational products \cite{moreira2015tandem}.

Our main contributions can be summarized as follows.
\begin{itemize}
    \item We establish the data model of single-look multi-master SAR tomography, by means of which both sparse recovery and model-order selection can be formulated as nonconvex minimization problems.
    \item We develop two algorithms for solving the aforementioned nonconvex sparse recovery problem, namely,
    \begin{enumerate}
        \item NLS:\\
        we provide two theorems regarding the critical points of its subproblem's objective function that also underlies model-order selection;
        \item bi-convex relaxation and alternating minimization (BiCRAM):\\
        we propose to sample its solution path for the purpose of automatic regularization parameter tuning, and we show empirically that a simple diagonal preconditioning can effectively improve convergence.
    \end{enumerate}
    \item We propose to correct quantization errors by using (nonconvex) nonlinear optimization.
    \item We validate tomographic height estimates with ground truth generated by SAR simulations and geodetic corrections.
\end{itemize}

The rest of this paper is organized as follows. Sec.~\ref{sec_multimaster} introduces the data model and inversion framework for single-look multi-master SAR tomography. In Sec.~\ref{sec_nonlinear} and \ref{sec_biconvex}, two algorithms for solving the nonconvex sparse recovery problem within the aforementioned framework, namely, NLS and BiCRAM, are elucidated and analyzed, respectively. Sec.~\ref{sec_experiment} reports an experiment with TerraSAR-X data including a validation of tomographic height estimates. This paper is concluded by Sec.~\ref{sec_conclusion}.

\section{Single-look multi-master SAR tomography} \label{sec_multimaster}

In this section, we establish the data model for single-look multi-master SAR tomography, analyze its implications for two specific cases, and sketch out a generic inversion framework for it.

We start with the mathematical notations that are used throughout this paper.
\begin{notation}
We denote scalars as lower- or uppercase letters (e.g., $m$, $N$, $\lambda$), vectors as bold lowercase letters (e.g., $\bfg$, $\bfgamma$), matrices, sets and ordered pairs as bold uppercase letters (e.g., $\bfR$, $\bfOmega$), and number fields as blackboard bold uppercase letters (e.g., $\intg$, $\real$, $\cplx$) with the following conventions:
\begin{itemize}
    \item $g_n$ denotes the $n$th entry of $\bfg$.
    \item $\bfa^m$ and $\bfa_n$ denote the $m$th row and $n$th column of $\bfA$, respectively.
    \item $\Diag(\bfa)$ denotes a square diagonal matrix whose entries on the main diagonal are equal to $\bfa$, and $\Diag(\bfA)$ denotes a vector whose entries are equal to those on the main diagonal of $\bfA$.
    \item $\Supp(\bfx)$ denotes the index set of nonzero entries or support of $\bfx$.
    \item $\conj{\bfA}$, $\bfA^T$ and $\bfA^H$ denote the (elementwise) complex conjugate, transpose and conjugate transpose of $\bfA$, respectively.
    \item $\bfA_R$ and $\Real(\bfA)$ denote the real part of $\bfA$.
    \item $\bfA_I$ and $\Imag(\bfA)$ denote the imaginary part of $\bfA$.
    \item $\bfA \had \bfB$ denotes the Hadamard product of $\bfA$ and $\bfB$.
    \item $\bfA \succ \bfzero$, $\bfB \prec \bfzero$ means that $\bfA$ is positive definite and $\bfB$ is negative definite.
    \item $\bfA_{\bfOmega}$ denotes the matrix formed by extracting the columns of $\bfA$ indexed by $\bfOmega$.
    \item $\|\bfA\|_{1,2}$ denotes the $\ell_{1,2}$ norm of $\bfA$, i.e., the sum of the $\ell_2$ norms of its rows.
    \item $\bfI$ denotes the identity matrix.
    \item $[N]$ denotes the set $\{1, \ldots, N\}$.
    \item $|\bfOmega|$ denotes the cardinality of the set $\bfOmega$.
    \item The nonnegative and positive subsets of a number field $\mathbb F$ are denoted as $\mathbb F_+$ and $\mathbb F_{++}$, respectively.
\end{itemize}
\end{notation}

\subsection{Data Model} \label{sec_multimaster_data}

First of all, we give a definition of ``single-master'' and ``multi-master'' by using the language of basic graph theory (e.g., \cite[\S1]{bondy2008graph}). Let $\bfG \isdef \left(\bfV(\bfG), \bfE(\bfG)\right)$ be an acyclic directed graph that is associated with an incidence function $\psi_\bfG$, where $\bfV(\bfG) \isdef [N]$ is a set of vertices (SLCs), $\bfE(\bfG)$ is a set of edges (interferograms), and for each $e \in \bfE(\bfG)$, $\exists m, n \in \bfV(\bfG)$ such that $\psi_\bfG(e)=(m, n)$. Its adjacency matrix $\bfA(\bfG) \isdef (a_{m,n}) \in \{0,1\}^{N \times N}$ is given by
\begin{equation}
    a_{m,n} \isdef
    \begin{cases}
    1 & : (m,n) \in \bfE(\bfG), \\
    0 & : (m,n) \notin \bfE(\bfG).
    \end{cases}
\end{equation}
Since $\bfG$ is acyclic, $a_{n,n}=0$, $\forall n \in \bfV(\bfG)$, i.e., the diagonal of $\bfA(\bfG)$ contains only zero entries. Without loss of generality, assume that every vertex is connected to at least another one.
\begin{definition}
The \emph{single-master} configuration means that there exists a unique $i \in [N]$ such that $a_{i,n} = 1$, $\bfa^m = \bfzero$, $\forall m,n \in [N] \setminus \{i\}$. In this case, we refer to $\{g_n \conj{g_i} / |g_i|\}$ as the \emph{single-master} stack with a master indexed by $i$.
\end{definition}
\begin{definition}
The \emph{multi-master} configuration means that $\nexists i \in [N]$ such that $a_{i,n} = 1$, $\bfa^m = \bfzero$, $\forall m,n \in [N] \setminus \{i\}$. In this case, we refer to $\{g_n \conj{g_m}\}$ as the \emph{multi-master} stack.
\end{definition}
That is, ``multi-master'' is equivalent to ``not single-master''. Two exemplary configurations are illustrated in Fig.~\ref{fig_single-multi-master_tsx_adj}.

\begin{figure}[!tp]
\centering
\includegraphics[width=0.48\textwidth]{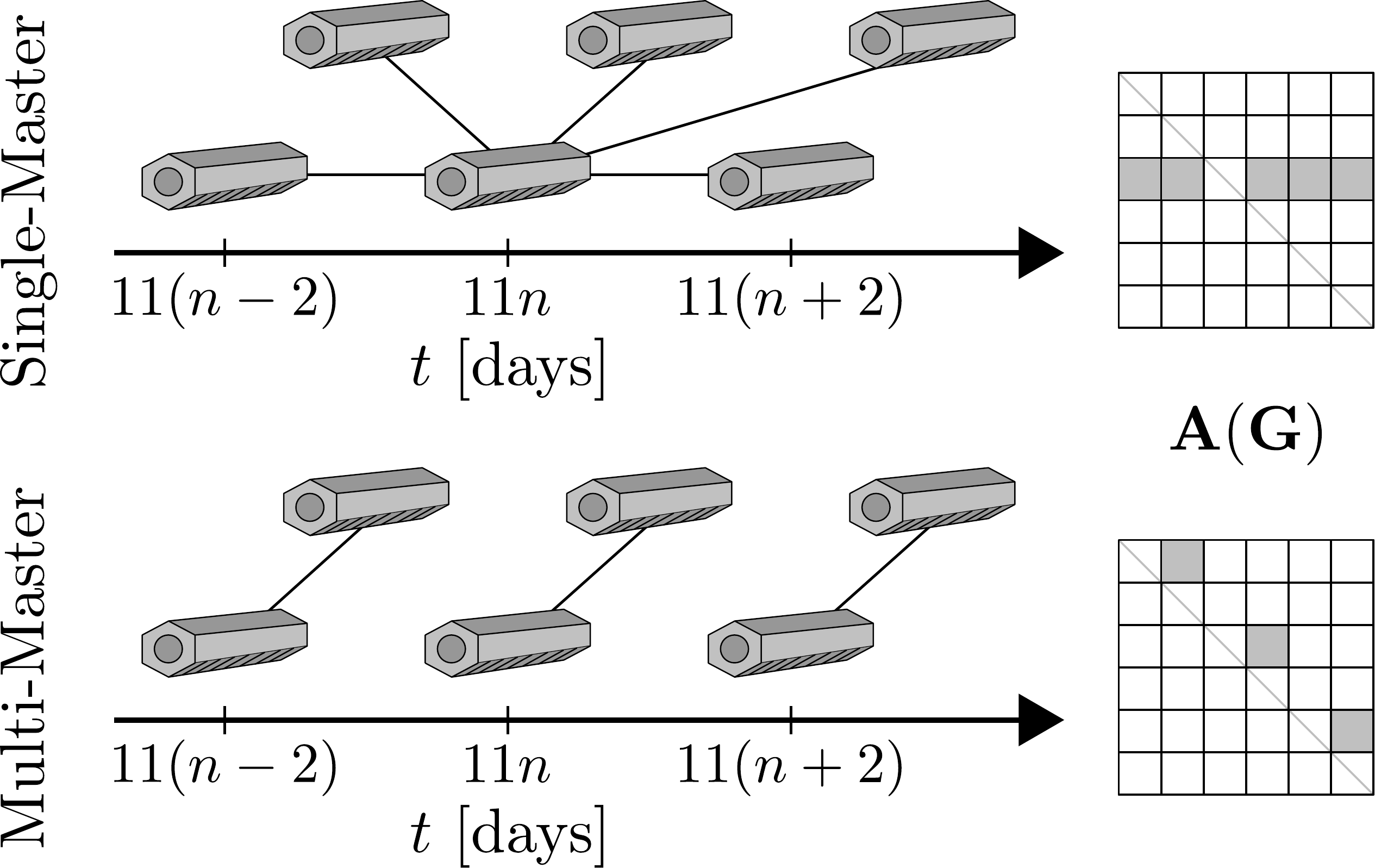}
\caption{Single-master vs.\ multi-master: two exemplary configurations and the corresponding adjacency matrices $\bfA(\bfG)$.} \label{fig_single-multi-master_tsx_adj}
\end{figure}

In the \emph{multi-master} case, an interferogram is created for each $(m,n) \in \bfE(\bfG)$:
\begin{equation} \label{eq_mm}
    g_n \conj{g_m} = \int \int \gamma(s) \conj{\gamma(s')} \exp\left( -j (k_ns - k_ms') \right) \mathrm ds \mathrm ds'.
\end{equation}

Hereafter, we focus on the case in which the far field contains only a small number of scatterers such that
\begin{equation}
    g_n \approx \sum_l \gamma_l \exp(-jk_ns_l), \quad n = 1, \ldots, N,
\end{equation}
where $\gamma_l \in \cplx$ is the reflectivity of the $l$th scatterer located at the elevation position $s_l$. The single-look multi-master data model \eqref{eq_mm} becomes
\begin{equation} \label{eq_mm_disc}
    g_n \conj{g_m} \approx \sum_{l, l'} \gamma_l \conj{\gamma_{l'}} \exp\left( -j (k_ns_l - k_ms_{l'}) \right),
\end{equation}
$\forall (m,n) \in \bfE(\bfG)$.


In the next subsection, we analyze the implications of \eqref{eq_mm_disc} for the single- and double-scatterer cases.

\subsection{Implications} \label{sec_multimaster_implications}

In the \emph{single-scatterer} case, \eqref{eq_mm_disc} becomes
\begin{equation} \label{eq_mm_single}
    g_n \conj{g_m} \approx |\gamma|^2 \exp\left( -j (k_n - k_m) s \right),
\end{equation}
i.e., the multi-master observation is actually the Fourier transform of the reflectivity \emph{power} at the wavenumber baseline $k_n-k_m$. As a result, the nonnegativity of $|\gamma|^2$ should be considered during inversion. Since both the real and imaginary parts of $g_n \conj{g_m}$ are parametrized by $|\gamma|^2$, i.e.,
\begin{equation}
\begin{aligned}
    \Real(g_n \conj{g_m}) &\approx |\gamma|^2 \cos\left( (k_n - k_m) s \right), \\
    \Imag(g_n \conj{g_m}) &\approx |\gamma|^2 \sin\left( -(k_n - k_m) s \right),
\end{aligned}
\end{equation}
the inversion problem can be recast as a real-valued one.

For \emph{double scatterers}, the multi-master observation is
\begin{equation} \label{eq_mm_double}
\begin{aligned}
    g_n \conj{g_m} \approx {}& |\gamma_1|^2 \exp\left( -j (k_n - k_m) s_1 \right) +{} \\
    & \gamma_1 \conj{\gamma_2} \exp\left( -j (k_ns_1 - k_ms_2) \right) +{} \\
    & \conj{\gamma_1} \gamma_2 \exp\left( -j (k_ns_2 - k_ms_1) \right) +{} \\
    & |\gamma_2|^2 \exp\left( -j (k_n - k_m) s_2 \right).
\end{aligned}
\end{equation}
In addition to the Fourier transform of the reflectivity power at $k_n-k_m$, the right-hand side of \eqref{eq_mm_double} contains the second and third ``cross-terms'' in which the reflectivity values of the two scatterers (and their frequency-time-products) are coupled. This essentially rules out any linear model.

\begin{remark}
In the \emph{multi-look} multi-master setting, the data model under random volume scattering is
\begin{equation}
    \E[\Gamma(k_n) \conj{\Gamma(k_m)}] = \int \sigma_0(s) \exp(-j (k_n-k_m) s) \mathrm ds,
\end{equation}
as already indicated in Eq.~\eqref{eq_ml_mm}, i.e., no coupling is involved.
\end{remark}

\begin{remark}
A \emph{multi-master} bistatic or pursuit monostatic (i.e., $10$-second temporal baseline \cite{hajnsek2014announcement}) stack is in general not motion-free for double (or multiple) scatterers.

To see this, consider for example the linear deformation model $d(t_n) \isdef vt_n$, where $v$ and $t$ denote linear deformation rate and temporal baseline, respectively. Observe that
\begin{equation}
\begin{aligned}
    & g_n \conj{g_m} \\
    \approx {}& \sum_{l, l'} \gamma_l \conj{\gamma_{l'}} \,\cdot \\
    & \exp( -j (k_ns_l - k_ms_{l'} + 4\pi d_l(t_n)/\lambda - 4\pi d_{l'}(t_m)/\lambda) ) \\
    = {}& |\gamma_1|^2 \exp\left( -j \left( (k_n - k_m) s_1 + 4\pi v_1(t_n-t_m)/\lambda \right) \right) +{} \\
    & \gamma_1 \conj{\gamma_2} \exp\left( -j \left( (k_ns_1 - k_ms_2) + 4\pi(v_1t_n-v_2t_m)/\lambda \right) \right) +{} \\
    & \conj{\gamma_1} \gamma_2 \exp\left( -j \left( (k_ns_2 - k_ms_1) + 4\pi(v_2t_n-v_1t_m)/\lambda \right) \right) +{} \\
    & |\gamma_2|^2 \exp\left( -j \left( (k_n - k_m) s_2 + 4\pi v_2(t_n-t_m)/\lambda \right) \right).
\end{aligned}
\end{equation}
In the case of $t_m=t_n$, the motion-induced phase in the cross-terms vanishes if and only if $v_1=v_2$.
\end{remark}

The next subsection introduces a generic inversion framework for single-look multi-master SAR tomography.

\subsection{Inversion Framework} \label{sec_multimaster_inversion}

The data model \eqref{eq_mm_disc} already indicates a nonlinear system of equations for a single-look multi-master stack. Suppose $\bfG$ is the graph associated with this stack that contains a total of $N' \isdef |\bfE(\bfG)|$ multi-master observations, and $e_1, \ldots, e_{N'}$ is an ordered sequence of all the edges in $\bfE(\bfG)$. Let $\Master, \Slave : [N'] \to [N]$ be the mappings to the master and slave image indices, respectively. For each $e_{n}$, $n \in [N']$, a multi-master observation $g_{\Slave(n)} \conj{g_{\Master(n)}}$ is obtained. Let $\bfg \in \cplx^{N'}$ be the vector of multi-master observations such that $g_{n} \isdef g_{\Slave(n)} \conj{g_{\Master(n)}}$, $\forall n \in [N']$. Let $s_1, \ldots, s_{L}$ be a discretization of the elevation axis $s$. The data model in matrix notations is
\begin{equation} \label{eq_mm_mat}
    \bfg \approx (\bfR \bfgamma) \had (\conj{\bfS \bfgamma}),
\end{equation}
where $\bfR, \bfS \in \cplx^{N' \times L}$ represent the tomographic observation matrices of the slave and master images, respectively, $r_{n,l} \isdef \exp(-j k_{\Slave(n)}s_l)$, $s_{n,l} \isdef \exp(-j k_{\Master(n)}s_l)$, $\forall n \in [N']$, $l \in [L]$, and $\bfgamma \in \cplx^L$ is the unknown reflectivity vector such that $\gamma_l$ is associated with the scatterer (if any) at elevation position $s_l$.

In light of \eqref{eq_mm_mat}, we propose the following framework for tomographic inversion.

\subsubsection{Nonconvex sparse recovery}
We consider the problem
\begin{equation} \label{eq_mm_sr}
\begin{aligned}
    \hat\bfgamma \isdef {}& \argmin_{\bfgamma} && \frac{1}{2} \| (\bfR \bfgamma) \had (\conj{\bfS \bfgamma}) - \bfg \|_2^2 \\
    & \st && |\Supp(\bfgamma)| \leq K,
\end{aligned}
\end{equation}
where $K\in\intg_{++}$. The objective function measures the model goodness of fit and the constraint enforces $\bfgamma$ to be sparse, as is implicitly assumed in \eqref{eq_mm_disc}. If $\sum_{l=0}^K \binom{L}{l}$ is small, \eqref{eq_mm_sr} can be solved heuristically by using the algorithms that will be developed in Sec.~\ref{sec_nonlinear}. Sec.~\ref{sec_biconvex} is dedicated to another algorithm that solves a similar problem based on bi-convex relaxation.

\subsubsection{Model-order selection}
This procedure removes outliers and therefore reduces false positive rate. By using for example the Bayesian information criterion (e.g., \cite{stoica2004model}), model-order selection can be formulated as the following constrained minimization problem:
\begin{equation}
\begin{aligned} \label{eq_mm_ms}
    \hat\bfOmega \isdef {}& \argmin_{\bfOmega(, \bfdelta)} && \, 2\ln\left( \| (\bfR_{\bfOmega} \bfdelta_{\bfOmega}) \had (\conj{\bfS_{\bfOmega} \bfdelta_{\bfOmega}}) - \bfg \|_2^2 / N'\right) \\
    & && {}+ (5|\bfOmega|+1) \ln(N')/N' \\
    & \st && \Supp(\bfdelta) = \bfOmega \subset \Supp(\hat\bfgamma),
\end{aligned}
\end{equation}
where $\bfdelta\in\cplx^L$ is an auxiliary variable, and $\bfOmega$ is its support. Since $|\Supp(\hat\bfgamma)|$ is typically small, \eqref{eq_mm_ms} can be tackled by solving a sequence of subset nonlinear least squares problems in the form of
\begin{equation} \label{eq_mm_ms_sub}
    \minimize_{\bfepsilon} \frac{1}{2} \| (\bfR_{\bfOmega} \bfepsilon) \had (\conj{\bfS_{\bfOmega} \bfepsilon}) - \bfg \|_2^2,
\end{equation}
for which two solvers will be introduced in Sec.~\ref{sec_nonlinear}.

\subsubsection{Off-grid correction}
The off-grid or quantization problem arises when scatterers are not located on the (regular) grid of discrete elevation positions $s_1, \ldots, s_{L}$. Ge et al.\ \cite{ge2018spaceborne} proposed to oversample $\hat\bfgamma$ in the vicinity of selected scatterers in order to circumvent this problem. Here we propose a more elegant approach that is based on nonlinear optimization.

Denote $\hat K \isdef |\hat\bfOmega|$ as the number of scatterers after model-order selection. Let $\gamma_l^R$ and $\gamma_l^I$ be the real and imaginary parts of the complex-valued reflectivity $\gamma_l$ of the $l$th scatterer that is located at $s_l$, respectively, i.e., $\gamma_l=\gamma_l^R+j\gamma_l^I$, $\forall l \in [\hat K]$. On the basis of the single-look multi-master data model \eqref{eq_mm_disc}, we seek a solution of the following minimization problem:
\begin{equation} \label{eq_mm_oc}
\begin{aligned}
    \minimize_{\gamma_l^R, \gamma_l^I, s_l} \sum_{n} \Big| & g_n - \sum_{l, l'} (\gamma_l^R+j\gamma_l^I) (\gamma_{l'}^R-j\gamma_{l'}^I) \,\cdot \\
    & \exp\left( -j (k_{\Slave(n)}s_l - k_{\Master(n)}s_{l'}) \right) \Big|^2.
\end{aligned}
\end{equation}
Note that the objective function is differentiable w.r.t.\ $\gamma_l^R$, $\gamma_l^I$ and $s_l$, $\forall l \in [\hat K]$. Needless to say, the on-grid estimates from \eqref{eq_mm_ms} are used as the initial solution. We will revisit this problem in Sec.~\ref{sec_nonlinear_algorithm}.

Thus far the inversion framework has been established. In the next two sections, we will deal with the optimization problems \eqref{eq_mm_sr}--\eqref{eq_mm_oc} from the algorithmic point of view.

\section{Nonlinear Least Squares (NLS)} \label{sec_nonlinear}

NLS is a parametric method that breaks down a sparse recovery problem into a series of subset linear least squares subproblems \cite[\S6.4]{stoica2005spectral}. Here we extend the concept of NLS to the single-look multi-master data model \eqref{eq_mm_mat} and address the subproblem \eqref{eq_mm_ms_sub}, or equivalently,
\begin{equation} \label{eq_nls}
    \minimize_{\bfx} \frac{1}{2} \| (\bfA \bfx) \had (\conj{\bfB \bfx}) - \bfb \|_2^2,
\end{equation}
where $\bfA,\bfB \in \cplx^{m \times n}$, $\bfx \in \cplx^n$, $\bfb \in \cplx^m$ with $m>n$. As can be concluded from Sec.~\ref{sec_multimaster_inversion}, \eqref{eq_nls} is clearly of interest, since it not only solves the nonconvex sparse recovery problem \eqref{eq_mm_sr}, but also underlies model-order selection \eqref{eq_mm_ms}.

\subsection{Algorithm} \label{sec_nonlinear_algorithm}

In this subsection, we develop two algorithms for solving \eqref{eq_nls}.

The first algorithm is based on the alternating direction method of multipliers (ADMM) \cite{boyd2011distributed}. ADMM solves a minimization problem by alternatively minimizing its augmented Lagrangian \cite[p.~509]{foucart2017mathematical}, in which the augmentation term is scaled by a penalty parameter $\rho\in\real_{++}$. A short recap can be found in Appendix~\ref{app_admm}. It converges under very general conditions with medium accuracy \cite[\S3.2]{boyd2011distributed}.

Now we consider \eqref{eq_nls} in its equivalent form:
\begin{equation}
\begin{aligned}
    & \minimize_{\bfx,\bfz} && \frac{1}{2} \| (\bfA \bfx) \had (\conj{\bfB \bfz}) - \bfb \|_2^2 \\
    & \st && \bfx - \bfz = \bfzero.
\end{aligned}
\end{equation}
This is essentially a bi-convex problem with affine constraint \cite[\S9.2]{boyd2011distributed}. Applying the ADMM update rules leads to Alg.~\ref{alg_nls_admm}. Note that both $\bfx$- and $\bfz$-updates boil down to solving linear least squares problems.
\begin{algorithm}[H]
\caption{An ADMM-based algorithm for solving \eqref{eq_nls}} \label{alg_nls_admm}
\begin{algorithmic}[1]
	\State \textbf{Input:} $\bfA$, $\bfB$, $\bfb$, $\bfz^{(0)}$, $\rho$
	\State \textbf{Initialize} \Let{\bfz}{\bfz^{(0)}}
	\State \textbf{Until} stopping criterion is satisfied, \textbf{Do}
	\State \quad \Let{\tilde\bfA}{ \Diag(\conj{\bfB \bfz}) \bfA }
	\State \quad \Let{\bfx}{ ( \tilde\bfA^H \tilde\bfA + \rho\bfI )^{-1} ( \tilde\bfA^H\bfb + \rho\bfz - \bfy ) }
	\State \quad \Let{\tilde\bfB}{ \Diag(\conj{\bfA \bfx}) \bfB }
	\State \quad \Let{\bfz}{ ( \tilde\bfB^H \tilde\bfB + \rho\bfI )^{-1} ( \tilde\bfB^H\conj{\bfb} + \rho\bfx + \bfy ) }
	\State \quad \Let{\bfy}{ \bfy + \rho(\bfx - \bfz) }
	\State \textbf{Output:} \(\bfz\)
\end{algorithmic}
\end{algorithm}

The second algorithm uses the trust-region Newton's method that exploits second-order information for solving general unconstrained nonlinear minimization problems \cite[\S4]{nocedal2006numerical}. The rationale behind this choice is to circumvent saddle points that cannot be identified by first-order information \cite{sun2015nonconvex}. In each iteration, a norm ball or ``trust region'' centered at the current iterate is adaptively chosen. If the second-order Taylor polynomial of the objective function is sufficiently good for approximation, a descent direction is found via solving a quadratically constrained quadratic minimization problem. Suppose $f : \real^n \to \real$ is the objective function, the subproblem at the iterate $\bfx \in \real^n$ is
\begin{equation} \label{eq_newton_tr}
\begin{aligned}
& \minimize_{\Delta\bfx} && f(\bfx) + \nabla f(\bfx)^T \Delta\bfx + \frac{1}{2} \Delta\bfx^T \nabla^2 f(\bfx) \Delta\bfx \\
& \st && \|\Delta\bfx\|_2 \leq r,
\end{aligned}
\end{equation}
where $\Delta\bfx \in \real^n$ is the search direction, $\nabla f$ and $\nabla^2 f$ denote the gradient and Hessian of $f$, respectively, and $r\in\real_{++}$ is the current trust region radius. By means of the Karush-Kuhn-Tucker (KKT) conditions for nonconvex problems, Nocedal and Wright \cite[\S4.3]{nocedal2006numerical} divided \eqref{eq_newton_tr} into several cases: in one case a one-dimensional (1-D) root-finding problem w.r.t.\ the dual variable is solved by using for example the Newton's method, while in the others the solutions are analytic. Since the technical details are quite overwhelming, we do not intend to provide an exposition here. Interested readers are advised to refer to \cite[\S4.3]{nocedal2006numerical}. It can be shown that the trust-region Newton's method converges to a critical point with high accuracy under general conditions \cite[p.~92]{nocedal2006numerical}.

By verifying the Cauchy-Riemann equations (e.g., \cite[p.~50]{freitag2006complex}), it is easy to show that the objective function of \eqref{eq_nls} is not complex-differentiable w.r.t.\ $\bfx$. In lieu of using Wirtinger differentiation that does not contain all the second-order information, we exploit the fact that the mapping $\bfx \mapsto (\bfx_R, \bfx_I)$ is isomorphic and let
\begin{equation} \label{eq_nls_real}
    f(\bfx_R,\bfx_I) \isdef \frac{1}{2} \| (\bfA \bfx) \had (\conj{\bfB \bfx}) - \bfb \|_2^2,
\end{equation}
where $f : \real^n \times \real^n \to \real$ is real-differentiable w.r.t.\ $\bfx_R$ and $\bfx_I$. Straightforward computations reveal its gradient as
\begin{equation} \label{eq_nls_real_grad}
    \nabla f(\bfx_R,\bfx_I) =
    \begin{pmatrix}
    \frac{\partial f}{\partial \bfx_R} \\
    \frac{\partial f}{\partial \bfx_I}
    \end{pmatrix}
    =
    \begin{pmatrix}
    \Real\left(\bfd\right) \\
    \Imag\left(\bfd\right)
    \end{pmatrix},
\end{equation}
where
\begin{equation} \label{eq_nls_real_grad_d}
\begin{aligned}
    \bfd \isdef {}& \bfA^H\left( ( (\bfA\bfx) \had (\conj{\bfB\bfx}) - \bfb ) \had (\bfB\bfx) \right) +{} \\
    & \bfB^H\left( ( (\conj{\bfA\bfx}) \had (\bfB\bfx) - \conj{\bfb} ) \had (\bfA\bfx) \right),
\end{aligned}
\end{equation}
and its Hessian as
\begin{equation} \label{eq_nls_real_hess}
\begin{aligned}
    \nabla^2 f(\bfx_R,\bfx_I) &=
    \begin{pmatrix}
    \frac{\partial^2 f}{\partial \bfx_R^2} & \frac{\partial^2 f}{\partial \bfx_R \partial \bfx_I} \\
    \frac{\partial^2 f}{\partial \bfx_I \partial \bfx_R} & \frac{\partial^2 f}{\partial \bfx_I^2}
    \end{pmatrix} \\
    &=
    \begin{pmatrix}
    \Real\left( \bfC + \bfD + \bfE \right) & -\Imag\left( \bfC - \bfD + \bfE \right) \\
    \Imag\left( \bfC + \bfD + \bfE \right) & \Real\left( \bfC - \bfD + \bfE \right)
    \end{pmatrix},
\end{aligned}
\end{equation}
where
\begin{equation}
\begin{aligned}
    \bfC &\isdef && \bfA^H \Diag\left( (\bfB\bfx) \had (\conj{\bfB\bfx}) \right) \bfA +{} \\
    & && \bfB^H \Diag\left( (\bfA\bfx) \had (\conj{\bfA\bfx}) \right) \bfB, \\
    \bfD &\isdef && \bfA^H \Diag\left( (\bfA\bfx) \had (\bfB\bfx) \right) \conj{\bfB} +{} \\
    & && \bfB^H \Diag\left( (\bfA\bfx) \had (\bfB\bfx) \right) \conj{\bfA}, \\
    \bfE &\isdef && \bfA^H \Diag\left( (\bfA\bfx) \had (\conj{\bfB\bfx}) - \bfb \right) \bfB +{} \\
    & && \bfB^H \Diag\left( (\conj{\bfA\bfx}) \had (\bfB\bfx) - \conj{\bfb} \right) \bfA.
\end{aligned}
\end{equation}
Note that $\bfd : \cplx^n \to \cplx^n$ and $\bfC,\bfD,\bfE : \cplx^n \to \cplx^{n \times n}$ are essentially functions of $\bfx$. Here we drop the parentheses in order to simplify notation. For the same purpose, we adopt the following convention:
\begin{equation}
    f(\bfx) \isdef f\left(\Real(\bfx),\Imag(\bfx)\right) = f(\bfx_R,\bfx_I).
\end{equation}

By using the first- and second-order information of \eqref{eq_nls_real}, \eqref{eq_nls} can be directly tackled by the trust-region Newton's method via solving a sequence of subproblems in the form of \eqref{eq_newton_tr}. For any optimal point $\bfx^\star$, the KKT condition is
\begin{equation} \label{eq_nls_kkt}
    \nabla f(\bfx^\star) = \bfzero \iff \bfd(\bfx^\star) = \bfzero.
\end{equation}

\begin{remark}
Likewise, the objective function of \eqref{eq_mm_oc} is real-differentiable w.r.t.\ $\gamma_l^R$, $\gamma_l^I$ and $s_l$, $\forall l \in [\hat K]$. Therefore, the trust-region Newton's method is directly applicable. Alternatively, first-order methods such as Broyden-Fletcher-Goldfarb-Shanno (BFGS, see for example \cite[\S6.1]{nocedal2006numerical} and the references therein) can also be used.
\end{remark}

Needless to say, it is not guaranteed that these algorithms always converge to a global minimum. We will demonstrate later in Sec.~\ref{sec_experiment} that the solutions are often good enough. Fig.~\ref{fig_nls_obj_history_admm_trn} shows typical convergence curves in the case of double scatterers (\#6 in Sec.~\ref{sec_experiment_experimental}). In order to generate this plot, we first let one algorithm run non-stop until it converged with very high precision. We then took this solution as an optimal point $\bfx^\star$ and compared the absolute difference of the objective value $|f(\bfx)-f(\bfx^\star)|$. Both ADMM and the trust-region Newton's method converged to the same solution (up to a constant phase angle, see Sec.~\ref{sec_nonlinear_analysis}), although it only took the latter less than $10$ iterations. Still, the former can be interesting due to the simplicity of its update rules (see Alg.~\ref{alg_nls_admm}). In Sec.~\ref{sec_experiment}, the latter will be used for demonstration purposes.

\begin{figure}[!tp]
\centering
\includegraphics[width=0.48\textwidth]{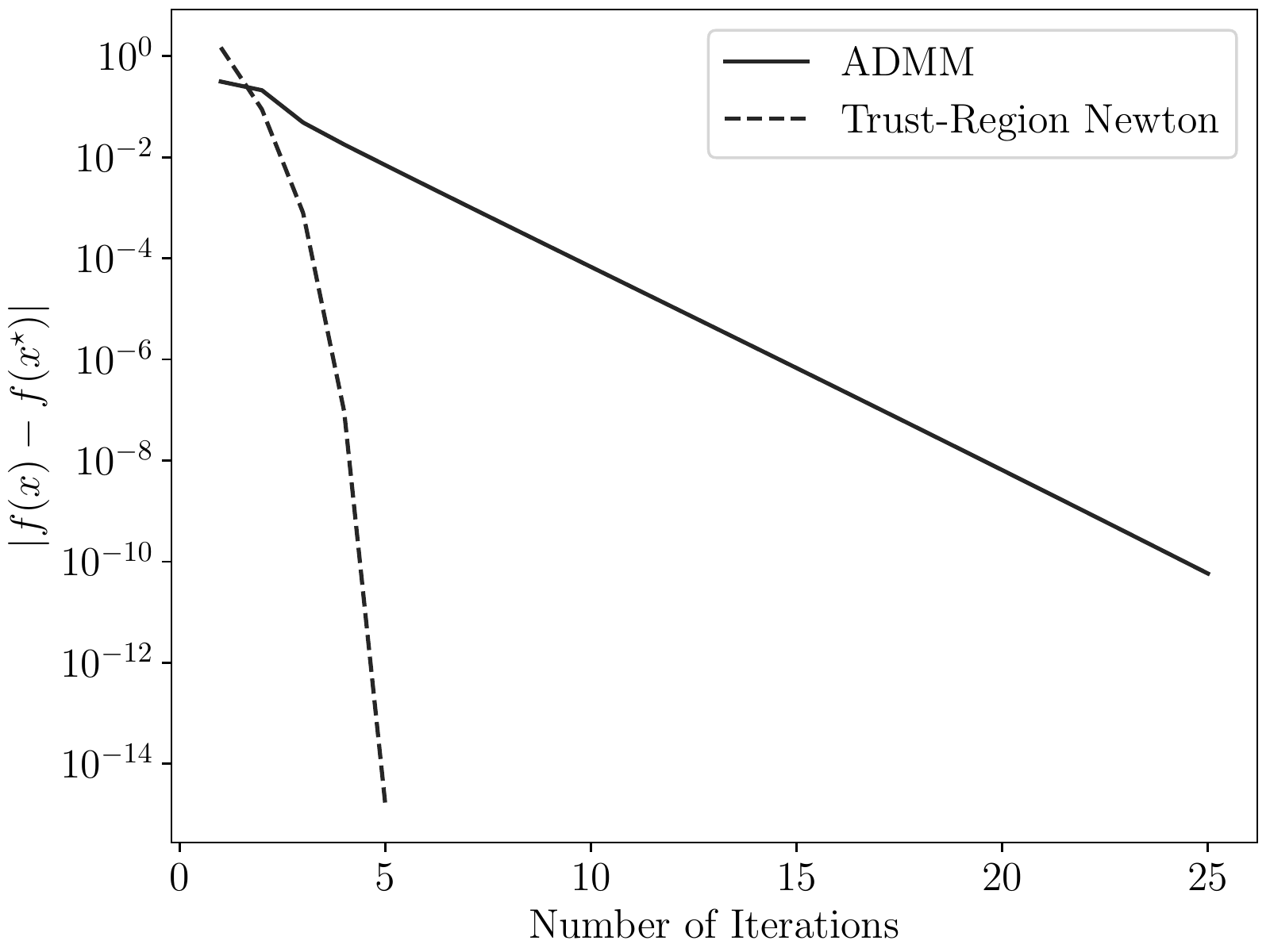}
\caption{Convergence curve of NLS using ADMM (solid line) and the trust-region Newton's method (dashed line).} \label{fig_nls_obj_history_admm_trn}
\end{figure}

\subsection{Analysis of the Objective Function} \label{sec_nonlinear_analysis}

Due to the nonconvexity of the objective function \eqref{eq_nls_real}, its analysis is not straightforward. We are primarily concerned with the following two questions:
\begin{enumerate}
    \item Under which circumstances do critical points or local extrema exist?
    \item If they do exist, how many are they?
\end{enumerate}
This subsection shall provide a partial answer to these questions.

First of all, we state the following general observation.
\begin{proposition} \label{thm_eig_hess}
For any $\bfx \in \cplx^n$ and $\phi \in \real$, any eigenvalue of $\nabla^2f(\bfx)$ is also an eigenvalue of $\nabla^2f \left( \bfx\exp(j\phi) \right)$ and vice versa.
\end{proposition}
\begin{proof}
See Appendix~\ref{app_prop}.
\end{proof}
Informally, this proposition implies that the definiteness of the Hessian is invariant under any rotation with a constant phase angle.

Now we state the main theorem for the general case.
\begin{theorem} \label{thm_general_case}
Properties of the critical points of $f(\bfx)$.
\begin{enumerate}[(a)]
    \item $\bfzero$ is a critical point: it is a local minimum if $\bfA^H\Diag(\bfb)\bfB + \bfB^H\Diag(\conj{\bfb})\bfA \prec \bfzero$, and a local maximum if $\bfA^H\Diag(\bfb)\bfB + \bfB^H\Diag(\conj{\bfb})\bfA \succ \bfzero$. \label{thm_general_case_zero}
    \item If there exists a \emph{nonzero} critical point, then $\bfA^H\Diag(\bfb)\bfB + \bfB^H\Diag(\conj{\bfb})\bfA \nprec \bfzero$. \label{thm_general_case_cond}
    \item Suppose there exists a \emph{nonzero} critical point $\bfz$. Then
    \begin{enumerate}[(1)]
        \item $\nabla^2f(\bfz)$ is rank deficient. \label{thm_general_case_rank}
        \item There exist an infinite number of critical points in the form of $\bfz\exp(j\phi)$, $\phi\in\real\setminus\{0\}$. Each has the same objective function value as $\bfz$, and its Hessian has the same definiteness. \label{thm_general_case_num}
    \end{enumerate}
\end{enumerate}
\end{theorem}
\begin{proof}
See Appendix~\ref{app_thm_general_case}.
\end{proof}
This theorem implies that if there exists one critical point, then there are an infinite number of them up to a constant phase angle, and each is as good. Furthermore, we conjecture that $\bfA^H\Diag(\bfb)\bfB + \bfB^H\Diag(\conj{\bfb})\bfA \succ \bfzero$ is a necessary and sufficient condition (cf.\ Thm.~\ref{thm_general_case}\eqref{thm_general_case_cond}), and each \emph{nonzero} critical point is also a local minimum under some mild conditions.

For the special case $n=1$, i.e., $\bfA,\bfB\in\cplx^m$, $\bfx\in\cplx$, we have a much stronger result.
\begin{theorem}[$n=1$] \label{thm_special_case}
Properties of the critical points of $f(\bfx)$.
\begin{enumerate}[(a)]
    \item $\bfzero$ is a critical point: it is a local minimum if $\Real\left( (\bfA\had\conj{\bfB})^H \bfb \right) < 0$, and a local maximum if $\Real\left( (\bfA\had\conj{\bfB})^H \bfb \right) > 0$. \label{thm_special_case_zero}
    \item There exists a \emph{nonzero} critical point if and only if $\Real\left( (\bfA\had\conj{\bfB})^H \bfb \right) > 0$. \label{thm_special_case_cond}
    \item Suppose there exists a \emph{nonzero} critical point $\bfz$. Then
    \begin{enumerate}[(1)]
        \item $\nabla^2f(\bfz)$ is positive semi-definite and rank deficient\footnote{Note that $\nabla^2f(\bfz) \in \real^{2 \times 2}$ by definition.}. \label{thm_special_case_rank}
        \item There exist an infinite number of critical points in the form of $\bfz\exp(j\phi)$, $\phi\in\real\setminus\{0\}$. Each has the same objective function value as $\bfz$, and its Hessian has the same definiteness.\label{thm_special_case_num}
        \item $\bfz$ is a local minimum. \label{thm_special_case_local}
    \end{enumerate}
\end{enumerate}
\end{theorem}
\begin{proof}
See Appendix~\ref{app_thm_special_case}.
\end{proof}
As a result, a \emph{nonzero local minimum} exists if and only if $\Real\left( (\bfA\had\conj{\bfB})^H \bfb \right) > 0$. If this condition is satisfied, then there are infinitely many local minima that are exactly as good. Fig.~\ref{fig_nls_obj_557} shows as an example the negative logarithm of \eqref{eq_nls_real} with a circle of local maxima.

\begin{figure}[!tp]
\centering
\includegraphics[width=0.48\textwidth]{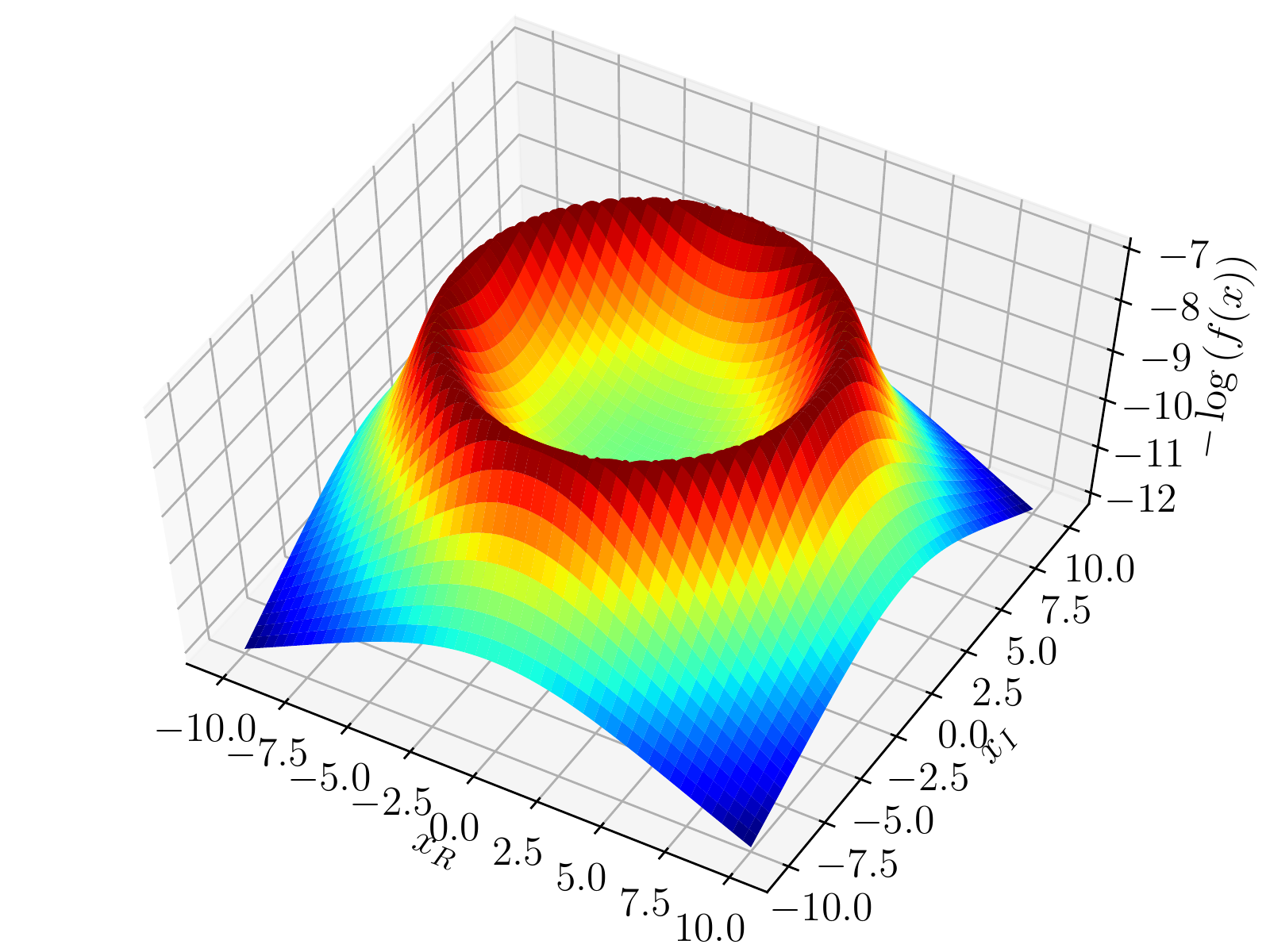} \\
\includegraphics[width=0.24\textwidth]{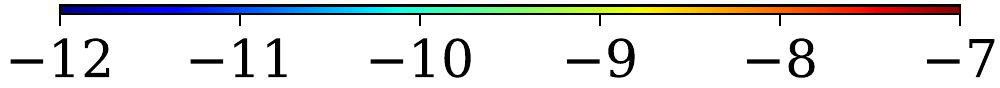}
\caption{Negative logarithm of the NLS objective function ($n=1$) with a circle of local maxima at the verge of the ``crater''.} \label{fig_nls_obj_557}
\end{figure}

Lastly, Thm.~\ref{thm_special_case} implies the following interesting result.
\begin{corollary}[$n=1$] \label{thm_special_case_crl}
Each \emph{nonzero local minimum} (if it exists) is given by
\begin{equation}
    \bfz = \frac{\Real\left( (\bfA\had\conj{\bfB})^H \bfb \right)^{1/2}}{\|\bfA\had\conj{\bfB}\|_2} \exp(j\phi),
\end{equation}
for some $\phi\in\real$.
\end{corollary}
\begin{proof}
See the proof of Thm.~\ref{thm_special_case}\eqref{thm_special_case_cond}.
\end{proof}

Now we return to our problem in SAR tomography. For the single-look multi-master data model \eqref{eq_mm_mat}, this corollary motivates the 1-D spectral estimator:
\begin{equation} \label{eq_nls_1d}
    |\hat\gamma_l| \isdef
    \begin{cases}
    \frac{\Real\left( (\bfr_l\had\conj{\bfs_l})^H \bfg \right)^{1/2}}{\|\bfr_l\had\conj{\bfs_l}\|_2} & \text{if } \Real\left( (\bfr_l\had\conj{\bfs_l})^H \bfg \right) > 0 \\
    0 & \text{otherwise},
    \end{cases}
\end{equation}
$\forall l \in [L]$. Note that this also provides the solution for any 1-D NLS subproblem up to a constant phase angle. In the case of multiple scatterers, this estimator does not have any super-resolution power.

\section{Bi-Convex Relaxation and Alternating Minimization (BiCRAM)} \label{sec_biconvex}

This section introduces a second algorithm for solving the nonconvex sparse recovery problem \eqref{eq_mm_sr}.

\subsection{Algorithm} \label{sec_biconvex_algorithm}

As a starting point, we replace the constraint in \eqref{eq_mm_sr} with a sparsity-inducing regularization term, e.g.,
\begin{equation} \label{eq_l1}
    \minimize_{\bfgamma} \frac{1}{2} \|(\bfR\bfgamma) \had (\conj{\bfS\bfgamma}) - \bfg\|_2^2 + \lambda \|\bfgamma\|_1,
\end{equation}
where $\lambda\in\real_{++}$ trades model goodness of fit for sparsity. In light of \eqref{eq_nls_kkt}, the necessary condition for being an optimal point $\bfgamma^\star$ is
\begin{equation}
\begin{aligned}
     \lambda\partial\|\bfgamma^\star\|_1 \ni {}& \bfR^H\left( ( \bfg - (\bfR\bfgamma^\star) \had (\conj{\bfS\bfgamma^\star}) ) \had (\bfS\bfgamma^\star) \right) +{} \\
    & \bfS^H\left( ( \conj{\bfg} - (\conj{\bfR\bfgamma^\star}) \had (\bfS\bfgamma^\star) ) \had (\bfR\bfgamma^\star) \right),
\end{aligned}
\end{equation}
i.e., the right-hand side is a subgradient of the $\ell_1$ norm at $\bfgamma^\star$. Obviously, $\bfzero$ always satisfies this condition.

In principle, an ADMM-based algorithm similar to Alg.~\ref{alg_nls_admm} can be used to solve \eqref{eq_l1}. However, our experience with real SAR tomographic data shows that it often diverges, presumably due to the high mutual coherence of $\bfR$ and $\bfS$ under nonconvexity. For this reason, we consider instead the following relaxed version of \eqref{eq_l1}:
\begin{equation} \label{eq_l1_relax}
\begin{aligned}
    & \minimize_{\bfgamma,\bftheta} && \frac{1}{2} \|(\bfR\bfgamma) \had (\conj{\bfS\bftheta}) - \bfg\|_2^2 + \frac{\lambda_1}{2} \|\bfgamma-\bftheta\|_2^2 +{} \\
    & && \lambda_2 \|
    \begin{pmatrix}
    \bfgamma & \bftheta
    \end{pmatrix}
    \|_{1,2},
\end{aligned}
\end{equation}
where $\lambda_1,\lambda_2\in\real_{++}$. The objective function $\cplx^L \times \cplx^L \to \real$ is bi-convex, i.e., it is convex in $\bfgamma$ with $\bftheta$ fixed, and convex in $\bftheta$ with $\bfgamma$ fixed. The first regularization term enforces $\bfgamma$ and $\bftheta$ to have similar entries, and the second one promotes the same support. Since \eqref{eq_l1_relax} is essentially an unconstrained bi-convex problem, it can be solved by using alternating minimization via Alg.~\ref{alg_bicram} (see also \cite{hong2018augmented, hong2019cospace, hong2019learnable}).
\begin{algorithm}[H]
\caption{An alternating algorithm for solving \eqref{eq_l1_relax}} \label{alg_bicram}
\begin{algorithmic}[1]
	\State \textbf{Input:} $\bfR$, $\bfS$, $\bfg$, $\bfgamma^{(0)}$, $\lambda_1$, $\lambda_2$
	\State \textbf{Initialize} \Let{\bfgamma}{\bfgamma^{(0)}}
	\State \textbf{Until} stopping criterion is satisfied, \textbf{Do}
	\State \quad \Let{\tilde\bfS}{ \Diag(\conj{\bfR \bfgamma}) \bfS }
	\State \quad \Let{\bftheta}{ \argmin_{\bftheta} \frac{1}{2} \|\tilde\bfS \bftheta - \conj{\bfg}\|_2^2 + \frac{\lambda_1}{2} \|\bftheta - \bfgamma\|_2^2 + \lambda_2 \| \begin{pmatrix} \bftheta & \bfgamma \end{pmatrix} \|_{1,2} }
	\State \quad \Let{\tilde\bfR}{ \Diag(\conj{\bfS \bftheta}) \bfR }
	\State \quad \Let{\bfgamma}{ \argmin_{\bfgamma} \frac{1}{2} \|\tilde\bfR \bfgamma - \bfg\|_2^2 + \frac{\lambda_1}{2} \|\bfgamma - \bftheta\|_2^2 + \lambda_2 \| \begin{pmatrix} \bfgamma & \bftheta \end{pmatrix} \|_{1,2} }
	\State \textbf{Output:} \(\bfgamma\)
\end{algorithmic}
\end{algorithm}
Each time when either $\bfgamma$ or $\bftheta$ is fixed, it becomes a convex problem in the generic form of:
\begin{equation} \label{eq_bicram}
    \minimize_{\bfx} \frac{1}{2} \|\bfA\bfx - \bfb\|_2^2 + \frac{\lambda_1}{2} \|\bfx - \bfu\|_2^2 + \lambda_2 \| \begin{pmatrix} \bfx & \bfu \end{pmatrix} \|_{1,2},
\end{equation}
or equivalently
\begin{equation} \label{eq_bicram_admm}
\begin{aligned}
    & \minimize_{\bfx,\bfZ} && \frac{1}{2} \|\bfA\bfx - \bfb\|_2^2 + \frac{\lambda_1}{2} \|\bfx - \bfu\|_2^2 + \lambda_2 \| \bfZ \|_{1,2} \\
    & \st && \begin{pmatrix} \bfx & \bfu \end{pmatrix} - \bfZ = \bfzero,
\end{aligned}
\end{equation}
where $\bfZ \in \cplx^{L \times 2}$. Applying the ADMM update rules leads to Alg.~\ref{alg_bicram_admm}.
\begin{algorithm}[H]
\caption{An ADMM-based algorithm for solving \eqref{eq_bicram}} \label{alg_bicram_admm}
\begin{algorithmic}[1]
	\State \textbf{Input:} $\bfA$, $\bfb$, $\bfu$, $\bfZ^{(0)}$, $\lambda_1$, $\lambda_2$, $\rho$
	\State \textbf{Initialize} \Let{\bfZ}{\bfZ^{(0)}}
	\State \textbf{Until} stopping criterion is satisfied, \textbf{Do}
	\State \quad \Let{\bfx}{ ( \bfA^H \bfA + (\lambda_1+\rho)\bfI )^{-1} ( \bfA^H\bfb + \lambda_1\bfu + \rho\bfz_1 - \bfy_1 ) }
	\State \quad \Let{\bfZ}{ \Prox_{\ell_{1,2}, \lambda_2/\rho} \left( \begin{pmatrix} \bfx & \bfu \end{pmatrix} + (1/\rho) \bfY \right) }
	\State \quad \Let{\bfY}{ \bfY + \rho \left( \begin{pmatrix} \bfx & \bfu \end{pmatrix} - \bfZ \right) }
	\State \textbf{Output:} \(\bfz_1\)
\end{algorithmic}
\end{algorithm}
$\Prox_{\ell_{1,2}, \lambda} : \cplx^{L \times 2} \to \cplx^{L \times 2}$ is the proximal operator of the $\ell_{1,2}$ norm scaled by $\lambda$ (e.g., \cite{parikh2014proximal}), i.e.,
\begin{equation}
    \Prox_{\ell_{1,2}, \lambda}(\bfX) \isdef \argmin_{\bfZ} \, \lambda \|\bfZ\|_{1,2} + \frac{1}{2} \| \bfX - \bfZ \|_F^2,
\end{equation}
whose $i$-th row is given by \cite[\S6.5.4]{parikh2014proximal}
\begin{equation}
    \Prox_{\ell_{1,2}, \lambda}(\bfX)^i = (1 - \lambda/\|\bfx^i\|_2)_+ \, \bfx^i,
\end{equation}
where $(x)_+ \isdef \max(x, 0)$. This proximal operator promotes (the columns of) $\bfZ$ to be jointly sparse and therefore $\bfx$ to share the same support with $\bfu$.

Due to the nonconvexity of \eqref{eq_l1_relax}, it is very difficult to establish a convergence guarantee for Alg.~\ref{alg_bicram} from a theoretical point of view. However, our experiments with real SAR tomographic data (see Sec.~\ref{sec_experiment}) show that it converges empirically. As an example, Fig.~\ref{fig_bicram_obj_history_230} depicts a convergence curve in the case of two scatterers that are closely located (\#6 in Sec.~\ref{sec_experiment_experimental}).

\begin{figure}[!tp]
\centering
\includegraphics[width=0.48\textwidth]{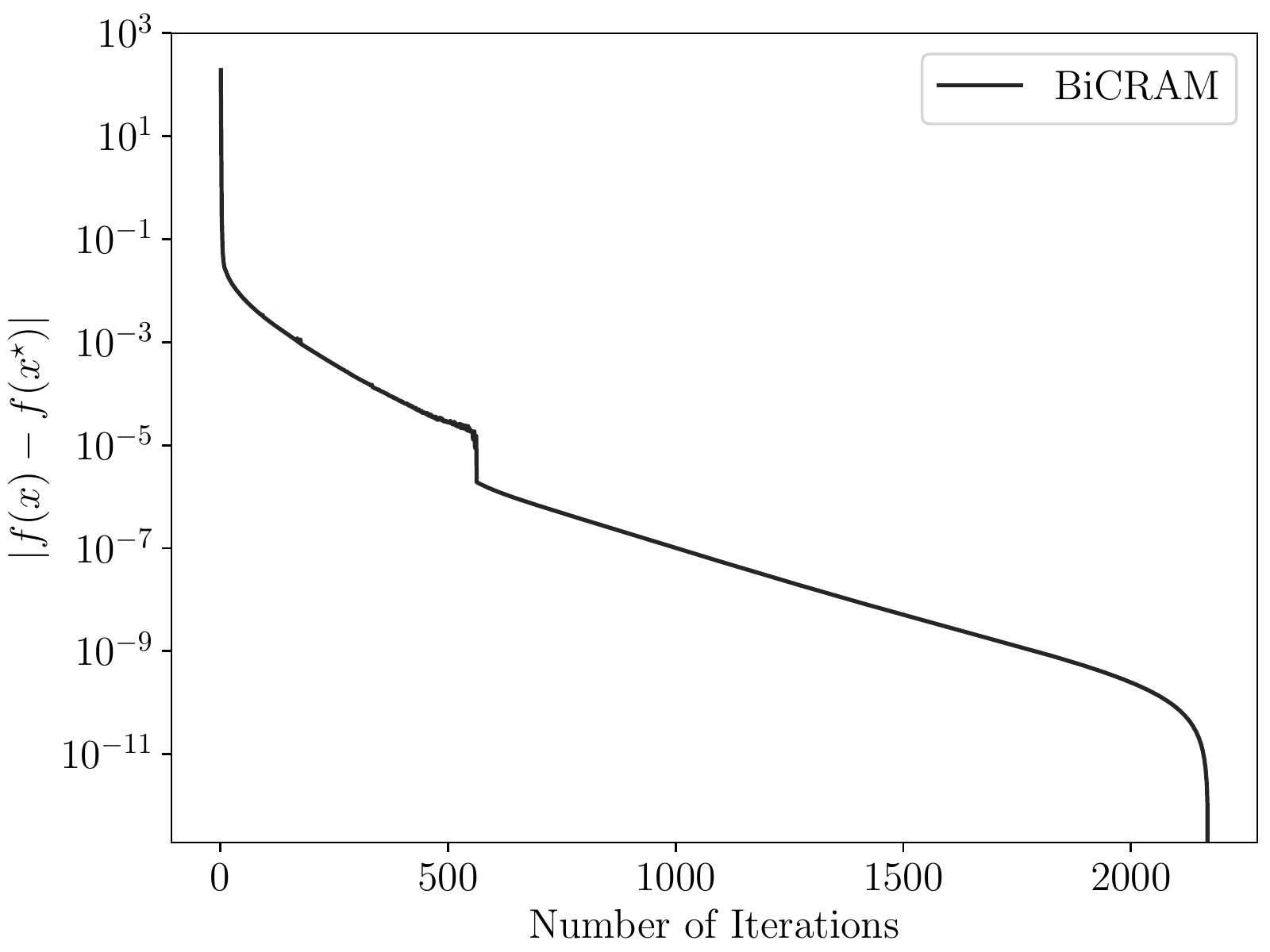}
\caption{Convergence curve of BiCRAM. The horizontal axis refers to the outer iterations in Alg.~\ref{alg_bicram}.} \label{fig_bicram_obj_history_230}
\end{figure}

In terms of regularization parameter tuning, we adopt the approach of sampling the solution path $(\lambda_1,\lambda_2) \mapsto \bfx$, and selecting the solution with the highest penalized likelihood \eqref{eq_mm_ms}. Last but not least, this procedure can be simplified by performing 1-D search, i.e., fixing one parameter and tuning the other at a time.

\subsection{Implementation} \label{sec_biconvex_implementation}

This subsection addresses several implementation aspects that contribute to accelerating Alg.~\ref{alg_bicram_admm} (and therefore Alg.~\ref{alg_bicram}). The exposition is based on an ADMM-based algorithm for solving the $\ell_1$-regularized least squares (L1RLS) problem:
\begin{equation} \label{eq_lasso}
    \minimize_{\bfx} \frac{1}{2} \|\bfA\bfx - \bfb\|_2^2 + \lambda \| \bfx \|_1.
\end{equation}
This is more suitable for demonstrating the power of different acceleration techniques, since each of its subproblems has an analytical solution and does not involve iteratively solving another optimization problem (cf.\ Alg.~\ref{alg_bicram}). Besides, it will also be used as a reference in Sec.~\ref{sec_experiment}.

Now consider \eqref{eq_lasso} in its equivalent form:
\begin{equation} \label{eq_lasso_admm}
\begin{aligned}
    & \minimize_{\bfx,\bfz} && \frac{1}{2} \|\bfA\bfx - \bfb\|_2^2 + \lambda \| \bfz \|_1 \\
    & \st && \bfx - \bfz = \bfzero.
\end{aligned}
\end{equation}
Applying the ADMM update rules leads to Alg.~\ref{alg_lasso_admm}.
\begin{algorithm}[H]
\caption{An ADMM-based algorithm for solving \eqref{eq_lasso}} \label{alg_lasso_admm}
\begin{algorithmic}[1]
	\State \textbf{Input:} $\bfA$, $\bfb$, $\bfz^{(0)}$, $\lambda$, $\rho$
	\State \textbf{Initialize} \Let{\bfz}{\bfz^{(0)}}
	\State \textbf{Until} stopping criterion is satisfied, \textbf{Do}
	\State \quad \Let{\bfx}{ ( \bfA^H \bfA + \rho\bfI )^{-1} ( \bfA^H\bfb + \rho\bfz - \bfy ) }
	\State \quad \Let{\bfz}{ \Prox_{\ell_1, \lambda/\rho} \left( \bfx + (1/\rho) \bfy \right) }
	\State \quad \Let{\bfy}{ \bfy + \rho \left( \bfx - \bfz \right) }
	\State \textbf{Output:} \(\bfz\)
\end{algorithmic}
\end{algorithm}
Likewise, $\Prox_{\ell_1, \lambda} : \cplx^L \to \cplx^L$ is the proximal operator of the $\ell_1$ norm scaled by $\lambda$ (also known as the soft thresholding operator \cite{donoho1995noising}):
\begin{equation}
    \Prox_{\ell_1, \lambda}(\bfx) \isdef \argmin_{\bfz} \, \lambda \|\bfz\|_1 + \frac{1}{2} \| \bfx - \bfz \|_2^2,
\end{equation}
whose $i$-th entry is given by \cite[\S6.5.2]{parikh2014proximal}
\begin{equation}
    \Prox_{\ell_1, \lambda}(\bfx)_i = (1 - \lambda/|x_i|)_+ \, x_i.
\end{equation}

The first technique provides an easier way for the $\bfx$-update.

\subsubsection{Matrix inversion lemma}
In Alg.~\ref{alg_bicram_admm} and \ref{alg_lasso_admm}, an $L$-by-$L$ matrix needs to be inverted. For large $L$, a direct exact approach can be tedious. Instead, we exploit the following lemma.
\begin{lemma}[Matrix inversion lemma \cite{boyd2004convex}] \label{thm_mat_inv}
For any $\bfA\in\cplx^{n \times m}$, $\bfB\in\cplx^{m \times n}$ and nonsingular $\bfC\in\cplx^{n \times n}$, we have
\begin{equation}
    (\bfA\bfB+\bfC)^{-1} = \bfC^{-1} - \bfC^{-1} \bfA (\bfI + \bfB \bfC^{-1} \bfA)^{-1} \bfB \bfC^{-1}.
\end{equation}
\end{lemma}
Lemma~\ref{thm_mat_inv} suggests a more efficient method if inverting $\bfC$ is straightforward. This is the case for matrices in the form of $\bfA^H\bfA+\rho\bfI$ since
\begin{equation} \label{eq_admm_mat_inv}
    (\bfA^H\bfA+\rho\bfI)^{-1} = \frac{1}{\rho}\bfI -\frac{1}{\rho^2} \bfA^H \left( \bfI + \frac{1}{\rho} \bfA \bfA^H \right)^{-1} \bfA,
\end{equation}
i.e., instead of the original $L$-by-$L$ matrix, only an $N'$-by-$N'$ matrix needs to be inverted.

Alternatively, the linear least squares (sub)problems can be solved iteratively in order to deliver an approximate solution \cite{fornasier2016conjugate}, which is known as inexact minimization \cite[\S3.4.4]{boyd2011distributed}.

The following techniques can be employed to improve convergence.

\subsubsection{Varying penalty parameter}
The penalty parameter $\rho$ can be updated at each iteration. Besides the convergence aspect, this also renders Alg.~\ref{alg_lasso_admm} less dependent on the initial choice of $\rho$. A common heuristic \cite[\S3.4.1]{boyd2011distributed} is to set
\begin{equation} \label{eq_admm_var_rho}
    \rho^{(k+1)} \isdef
    \begin{cases}
    \tau\rho^{(k)} & \text{if } \|\bfr^{(k)}\|_2 > \mu\|\bfs^{(k)}\|_2 \\
    \rho^{(k)}/\tau & \text{if } \|\bfs^{(k)}\|_2 > \mu\|\bfr^{(k)}\|_2 \\
    \rho^{(k)} & \text{otherwise}
    \end{cases}
\end{equation}
at the $(k+1)$th iteration, where $\tau,\mu>1$ are parameters, $\bfr^{(k)} \isdef \bfx^{(k)} - \bfz^{(k)}$ is the primal residual, and $\bfs^{(k)} \isdef \rho^{(k)} (\bfz^{(k)} - \bfz^{(k-1)})$ is the dual residual. As $k \to \infty$, $\bfr^{(k)}$ and $\bfs^{(k)}$ both converge to $\bfzero$. Intuitively, increasing $\rho$ tends to put a larger penalty on the augmenting term $(\rho/2) \|\bfx-\bfz\|_2^2$ in the augmented Lagrangian and consequently decrease $\|\bfr^{(k)}\|_2$ on the one hand, and to increase $\|\bfs^{(k)}\|_2$ by definition on the other and vice versa. The rationale is to balance $\bfr^{(k)}$ and $\bfs^{(k)}$ so that they are approximately of the same order. Naturally, one downside is that \eqref{eq_admm_mat_inv} needs to be recomputed whenever $\rho$ changes.

\subsubsection{Diagonal preconditioning}
The augmenting term $(\rho/2) \|\bfx-\bfz\|_2^2$ in the augmented Lagrangian can be replaced by
\begin{equation}
    (1/2) \langle \bfP(\bfx-\bfz), \bfx-\bfz \rangle,
\end{equation}
where $\bfP\succ\bfzero$ is a real diagonal matrix. Note that this falls under the category of more general augmenting terms \cite[\S3.4.2]{boyd2011distributed}. By means of this, Alg.~\ref{alg_lasso_admm} is deprived of the burden of choosing $\rho$ and the ADMM updates become
\begin{equation}
\begin{aligned}
    \bfx & \leftarrow ( \bfA^H \bfA + \bfP )^{-1} ( \bfA^H\bfb + \bfP\bfz - \bfy ) \\
    \bfz & \leftarrow \Prox_{\ell_1, \lambda/\bfp} \left( \bfx + \bfP^{-1} \bfy \right) \\
    \bfy & \leftarrow \bfy + \bfP \left( \bfx - \bfz \right),
\end{aligned}
\end{equation}
where $\bfp\isdef\Diag(\bfP)$, and $\Prox_{\ell_1, \bfw} : \cplx^L \to \cplx^L$ is the proximal operator of the weighted $\ell_1$ norm with weights $\bfw \in \real_{++}^L$:
\begin{equation}
    \Prox_{\ell_1, \bfw}(\bfx) \isdef \argmin_{\bfz} \, \|\bfz\|_{\bfw,1} + \frac{1}{2} \| \bfx - \bfz \|_2^2,
\end{equation}
whose $i$-th entry is given by
\begin{equation}
    \Prox_{\ell_1, \bfw}(\bfx)_i = (1 - w_i/|x_i|)_+ \, x_i.
\end{equation}
In case $\bfA^H\bfA$ is ill-conditioned (such as in SAR tomography), $\bfP$ can be interpreted as a preconditioner. Needless to say, Lemma~\ref{thm_mat_inv} can be also applied to invert $\bfA^H\bfA+\bfP$.

Pock and Chambolle (2011) proposed a simple and elegant way to construct diagonal preconditioners for a primal-dual algorithm \cite{chambolle2011first} \cite[\S15.2]{foucart2017mathematical} with guaranteed convergence:
\begin{equation}
    p_i \isdef 1 / \| \bfa_i \|_\alpha^\alpha, \quad \forall i \in [L],
\end{equation}
where $\alpha\in[0,2]$ is a parameter.

\subsubsection{Over-relaxation}
This means inserting between the $\bfx$- and $\bfz$-updates of Alg.~\ref{alg_lasso_admm} the following additional update:
\begin{equation}
    \bfx \leftarrow \beta\bfx  + (1-\beta)\bfz,
\end{equation}
where $\beta\in[1.5,1.8]$ (see for example \cite[\S3.4.3]{boyd2011distributed} and the references therein).

Fig.~\ref{fig_lasso_obj_history_admm_230} shows the convergence curve of Alg.~\ref{alg_lasso_admm} using different acceleration techniques, as applied to real SAR tomographic data (\#6 in Sec.~\ref{sec_experiment}). Each technique did contribute to accelerating Alg.~\ref{alg_lasso_admm} in comparison to ``baseline'', where we set $\rho=1$. The number of iterations of this and five other cases are listed in Tab.~\ref{tab_iter_alg_lasso}. Obviously, the combination of diagonal preconditioning and over-relaxation was the most competitive one and will therefore be adopted for all the ADMM-based algorithms in the following.

\begin{figure}[!tp]
\centering
\includegraphics[width=0.48\textwidth]{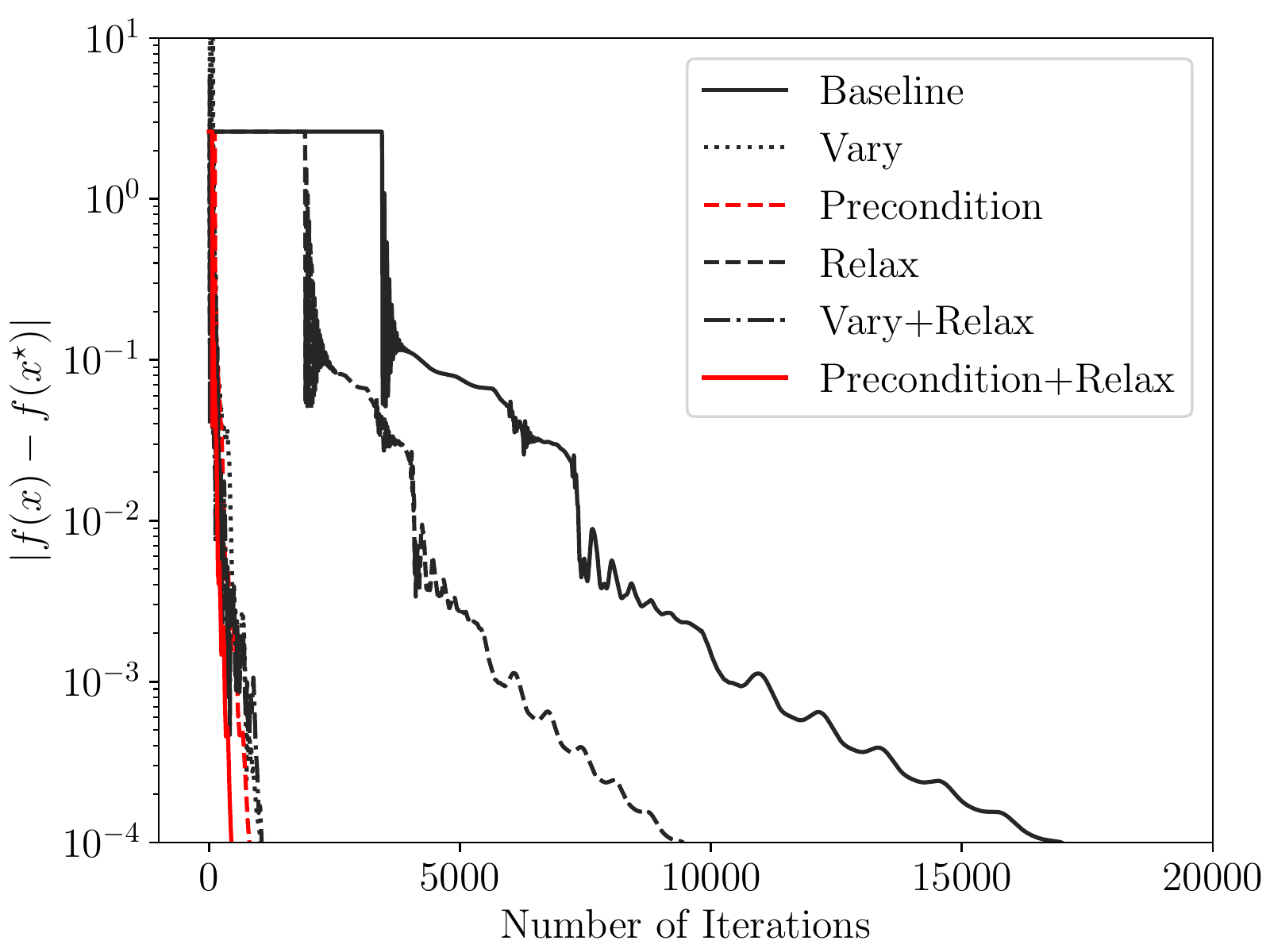} \\
\includegraphics[width=0.48\textwidth]{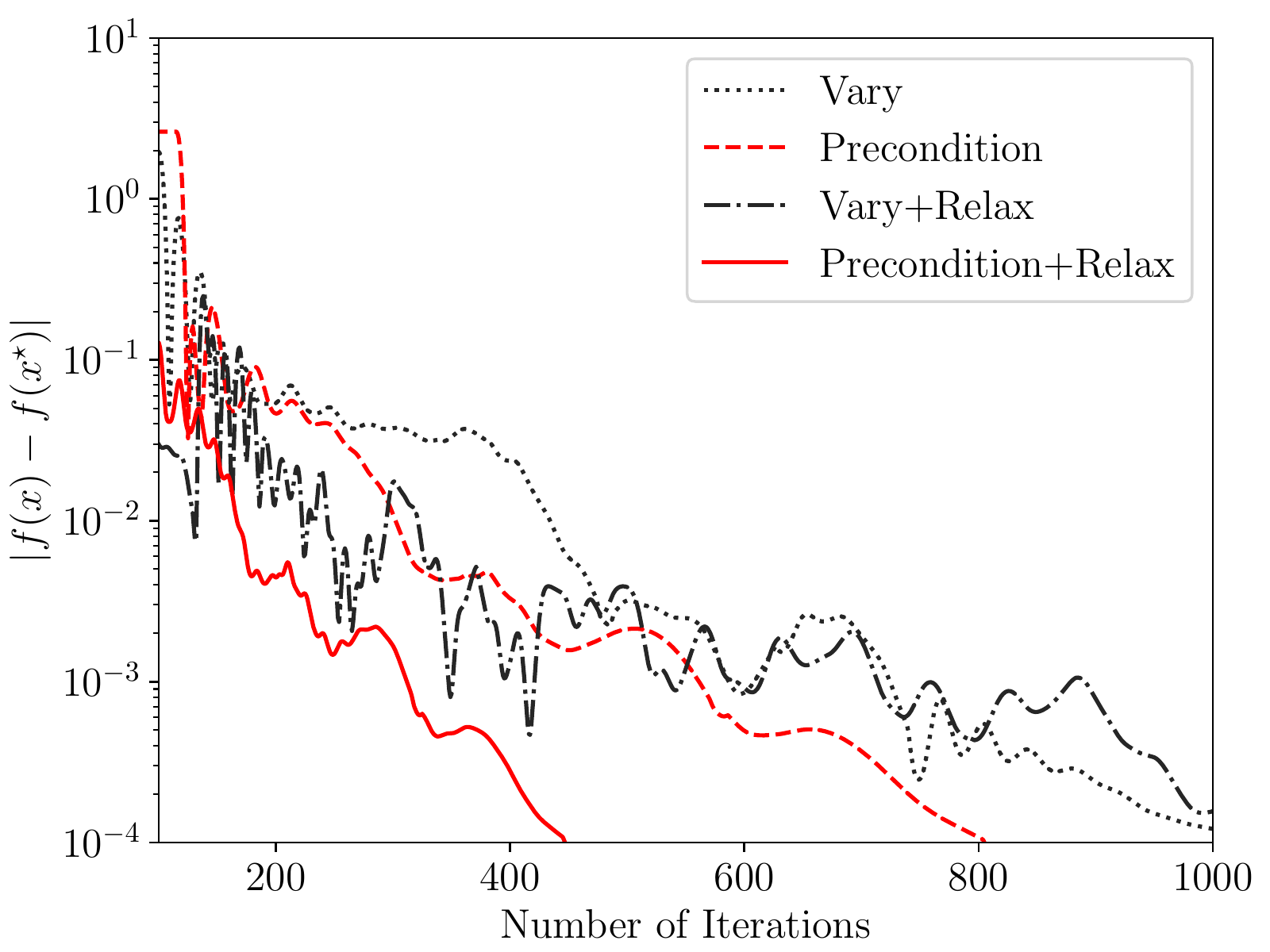}
\caption{Convergence curve of Alg.~\ref{alg_lasso_admm} using different acceleration techniques. ``Baseline'': $\rho=1$. ``Vary'': varying penalty parameter. ``Precondition'': diagonal preconditioning. ``Relax'': over-relaxation. Bottom: a close-up of the top figure.} \label{fig_lasso_obj_history_admm_230}
\end{figure}

\begin{table*}[!t]
    \renewcommand{\arraystretch}{1.3}
    \caption{Number of iterations using different acceleration techniques}
    \label{tab_iter_alg_lasso}
    \centering
    \begin{tabular}{l r r r r r r}
        & \#1 & \#2 & \#3 & \#4 & \#5 & \#6 \\
        \hline
        Baseline & $20289$ & $14250$ & $15752$ & $19571$ & $23920$ & $16983$ \\
        Vary & $1445$ & $1101$ & $1780$ & $1958$ & $2005$ & $1033$ \\
        Precondition & $679$ & $480$ & $504$ & $1074$ & $840$ & $805$ \\
        Relax & $11272$ & $7916$ & $8750$ & $10866$ & $13288$ & $9446$ \\
        Vary+Relax & $767$ & $625$ & $871$ & $821$ & $834$ & $1051$ \\
        Precondition+Relax & $\underline{377}$ & $\underline{260}$ & $\underline{288}$ & $\underline{595}$ & $\underline{467}$ & $\underline{447}$ \\
        \hline
    \end{tabular}
\end{table*}

\section{Experiment with TerraSAR-X Data} \label{sec_experiment}

In this section, we report our experimental results with a real SAR data set.

\subsection{Design of Experiment} \label{sec_experiment_design}

As a demonstration, we used $31$ TerraSAR-X staring spotlight repeat-pass acquisitions of the central Munich area from March 31, 2016 to December 7, 2017. This data set was processed with DLR's Integrated Wide Area Processor \cite{adam2013wide, gonzalez2013integrated}, as was elaborately described in \cite{ge2018spaceborne}. In addition to side lobe detection (see \cite{ge2018spaceborne} and the references therein), any non-peak point inside a main lobe was also removed, since it would otherwise lead to a ``ghost'' scatterer in the result, as any side lobe point would do too. Our region of interest contains a six-story building (``Nordbau'') of the Technical University of Munich (TUM) shown in Fig.~\ref{fig_tum-nordbau_photo_model} (left). The building signature in the SAR intensity image can be observed in Fig.~\ref{fig_mean_ampl_double_scatterer}, where the regular grid of salient points within the building footprint is a result of triple reflections on three orthogonal surfaces: metal plate (behind window glass), window ledge and brick wall \cite{auer2019sigs}. After main and side lobe detection, a total of $594$ looks were left, whose azimuth-range positions are shown in Fig.~\ref{fig_ps_x_y_sim_h} (bottom).

A \emph{single-master} stack was formed by choosing the acquisition from December 20, 2016 as the one and only master. Its vertical wavenumbers are shown in Fig.~\ref{fig_stack_wavenumber}. A sinusoidal basis function was used for modeling periodical motion induced by temperature change. The vertical Rayleigh resolution at scene center is approximately $12.66$~m. The Cr\'amer-Rao lower bound (CRLB) of height estimates given the aforementioned periodical deformation model \cite{ge2019bistatic} and a nominal signal-to-noise ratio (SNR) of $2$~dB is approximately $1.10$~m. NLS and L1RLS were applied to this stack for tomographic reconstruction. The latter was solved by Alg.~\ref{alg_lasso_admm} augmented with diagonal preconditioning and over-relaxation (see Sec.~\ref{sec_biconvex_implementation}), where we set $\beta=1.8$ and the choice of $\alpha$ is irrelevant (since $\bfA$ is a Fourier matrix). The solution path of L1RLS was sampled $11$ times with the regularization parameter varying logarithmically from $\lambda_{\min} \isdef 5 \cdot 10^{-2} \|\bfR^H\bfg\|_{\infty}$ to $\lambda_{\max} \isdef 5 \cdot 10^{-1} \|\bfR^H\bfg\|_{\infty}$.

We constructed a \emph{multi-master} stack of small temporal baselines: suppose $1', 2', 3', 4', \ldots$ is a chronologically ordered sequence of SLCs, the interferograms (edges) are $(1',2')$, $(3',4')$, etc. (see Fig.~\ref{fig_single-multi-master_tsx_adj}). As a result, this stack consists of $15$ interferograms. Due to the small-baseline feature of this stack, we did not employ any deformation model for the sake of simplicity. NLS and BiCRAM were applied to reconstruct the elevation profile, where the latter was solved by Alg.~\ref{alg_bicram_admm} employing diagonal preconditioning and over-relaxation. Likewise, the solution path of BiCRAM was also sampled $11$ times, where $\lambda_1$ was fixed as one (since it was deemed relatively insignificant as far as our experience went), and $\lambda_2$ was set to vary logarithmically from $\lambda_{\min} \isdef 5 \cdot 10^{-2} \max \left\{ \|\bfR^H\bfg\|_{\infty}, \|\bfS^H\bfg\|_{\infty} \right\}$ to $\lambda_{\max} \isdef 5 \cdot 10^{-1} \max \left\{ \|\bfR^H\bfg\|_{\infty}, \|\bfS^H\bfg\|_{\infty} \right\}$. The initial solution was given by $\bfgamma^{(0)} = (\bfR \had \conj{\bfS})^H \bfg$ due to its simplicity. Alternatively, \eqref{eq_nls_1d} could be used. In terms of off-grid correction, forward-mode automatic differentiation \cite{revels2016forward} was employed in order to circumvent analytically differentiating the objective function of \eqref{eq_mm_oc} for any number of scatterers, and the optimization problem was solved by means of a BFGS implementation \cite{mogensen2018optim}.

Finally, we built a second small-baseline \emph{multi-master} stack in the identical way as the previous one. In addition, we normalized each interferogram with the corresponding master amplitude. We will refer to this as the \emph{fake single-master} stack, since we treated it as if it had been a \emph{single-master} one. In order to apply the \emph{single-master} approach, we calculated for each interferogram the difference between slave and master wavenumbers, and used it as if it had been the wavenumber baseline, i.e., by inadequately assuming
\begin{equation}
    g_n \conj{g_m} \approx \sum_l \gamma_l \exp\left( -j (k_n - k_m) s_l \right),
\end{equation}
for each $(m,n) \in \bfE(\bfG)$. Needless to say, NLS and L1RLS were employed exactly the same as in the single-master case.

\begin{figure}[!tp]
\centering
\includegraphics[width=0.48\textwidth]{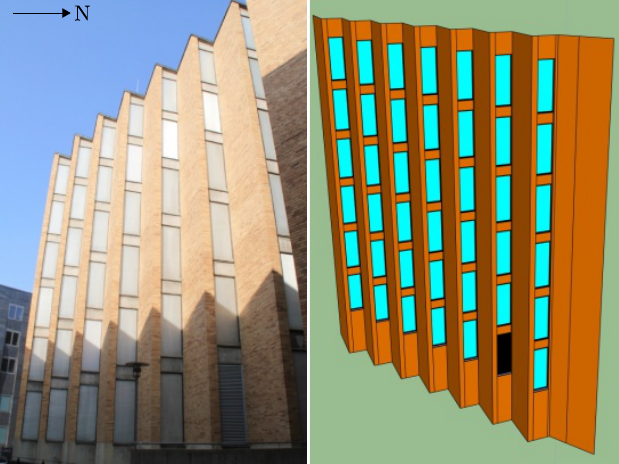}
\caption{Eastern facade of the six-story TUM-Nordbau building in our region of interest \cite{auer2019sigs}. Left: in-situ photo. Right: 3-D facade model. The black shape corresponds to a metallic window.} \label{fig_tum-nordbau_photo_model}
\end{figure}

\begin{figure}[!tp]
\centering
\includegraphics[width=0.48\textwidth]{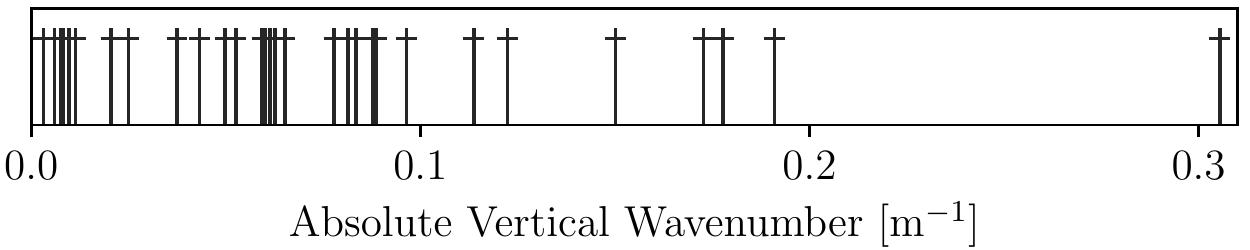}
\caption{Single-master absolute vertical wavenumbers. The largest one is approximately $0.31$ m$^{-1}$.} \label{fig_stack_wavenumber}
\end{figure}

The next subsection briefly explains how we generated ground truth data.

\subsection{Generation of Height Ground Truth} \label{sec_experiment_generation}

Height ground truth data was made available via a SAR imaging geodesy and simulation framework \cite{auer2019sigs}. The starting point was to create a three-dimensional (3-D) facade model from terrestrial measurements, via (drone-borne) camera, tachymeter, measuring rod and differential global positioning system (GPS), with an overall accuracy better than $2$~cm and a very high level of details \cite{auer2019sigs}. Ground control points were used for referencing this facade model to an international terrestrial reference frame. A visualization of the 3-D facade model is provided in Fig.~\ref{fig_tum-nordbau_photo_model} (right). The ray-tracing-based RaySAR simulator \cite{auer20113d} was employed to simulate dominant scatterers that, as already mentioned in Sec.~\ref{sec_experiment_design}, correspond to triple reflections on the building facade. With the help of atmospheric and geodynamic corrections from DLR's SAR Geodetic Processor \cite{eineder2010imaging, gisinger2014precise} and the newly enhanced TerraSAR-X orbit products \cite{hackel2018long}, their absolute coordinates were converted into azimuth timing, range timing and height that we refer to as \emph{Level~0} ground truth data.

\emph{Level~1} ground truth data consists of height at $30$ simulated PSs that are matched with real ones. The matching was conducted in the azimuth-range geometry, so as not to be affected by any height estimate error \cite{gernhardt2015persistent}. Fig.~\ref{fig_ps_x_y_sim_h} (top) shows the height simulations at the sub-pixel azimuth-range positions of the corresponding $30$ PSs. This height is relative to a corner reflector that is located on top of a neighboring TUM building and next to a permanent GPS station \cite{auer2019sigs}.

In addition, we performed height interpolation for a total of $594$ looks (see Sec.~\ref{sec_experiment_design}) in the following way. First of all, the height of each simulated PS was converted into interferometric phase. Next, the distance to the polyline representing the nearest-range cross-section of the building facade was used as the independent variable to construct a 1-D interpolator. In the end, the phase was interpolated at the previously mentioned $594$ looks and converted back into height. This interpolated height is referred to as the \emph{Level~$2$} ground truth and shown in Fig.~\ref{fig_ps_x_y_sim_h} (bottom). Needless to say, one assumption is that each scatterer, if it does exist, should lie on the building facade. A cross-validation of this 1-D interpolator was performed in \cite{auer2019sigs}, where the standard deviation (SD) and median absolute deviation (MAD) were shown to be $0.004$ and $0.002$~m, respectively.

\begin{figure}[!tp]
\centering
\includegraphics[width=0.4\textwidth]{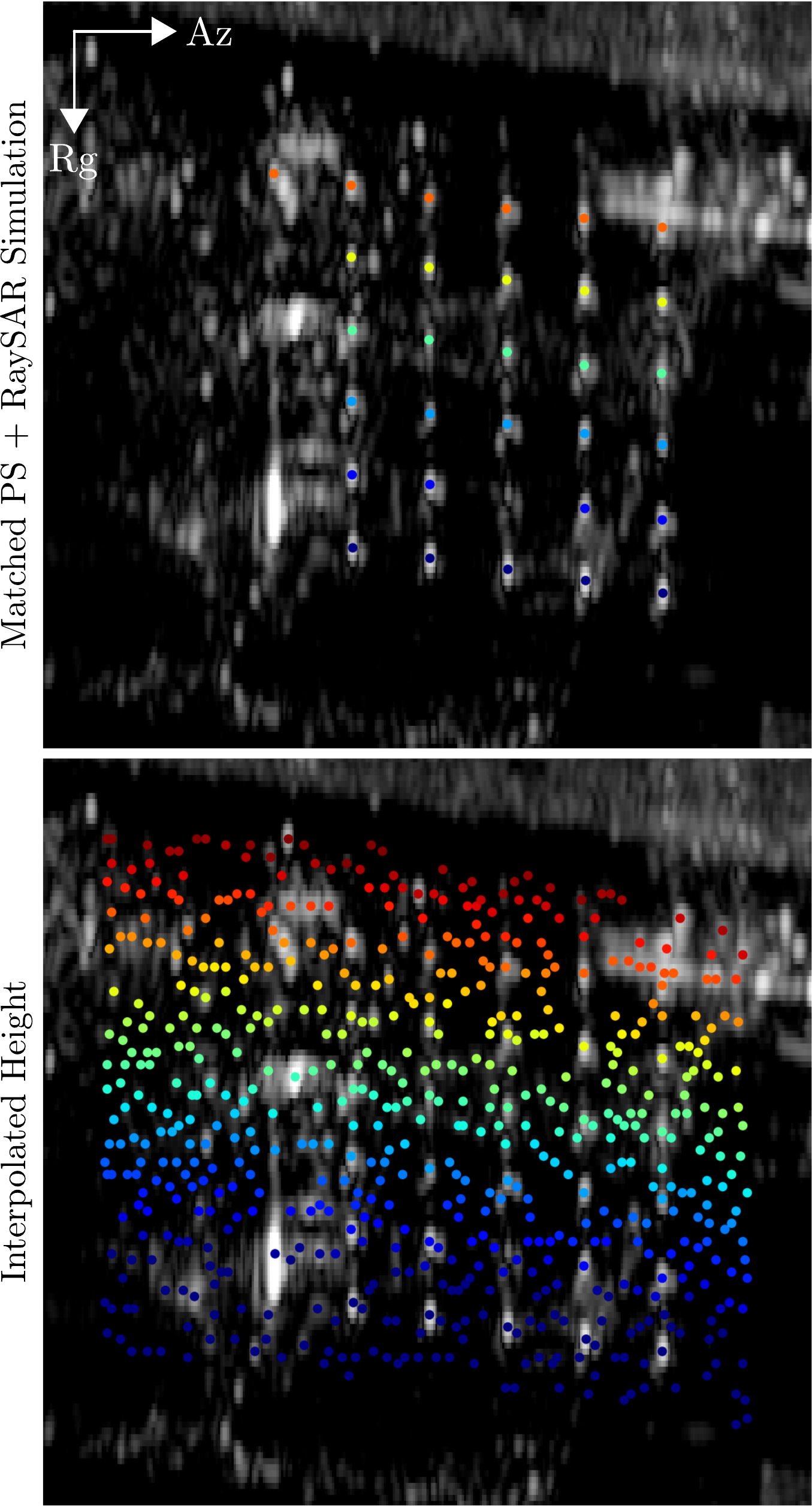} \\
\includegraphics[width=0.35\textwidth]{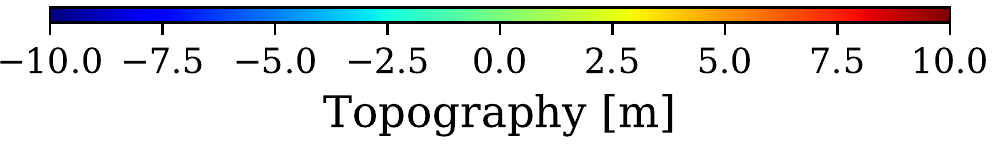}
\caption{SAR intensity image and height ground truth of our region of interest. Top: RaySAR height simulations at $30$ matched PS coordinates (Level~1). Bottom: interpolated height at $594$ facade looks (Level~2).} \label{fig_ps_x_y_sim_h}
\end{figure}

In the next subsection, our preliminary results are reported.

\subsection{Experimental Results} \label{sec_experiment_experimental}

The experiments can be divided into three categories: single-master, multi-master and fake single-master (see Sec.~\ref{sec_experiment_design}). In each category, two algorithms were applied for tomographic reconstruction.

As a proof of concept, we selected six looks that are very likely subject to facade-roof layover. These six looks were chosen in a systematic way: we performed tomographic reconstruction on the \emph{single-master} stack by using Tikhonov regularization (i.e., the $\ell_1$ norm in the regularization term of \eqref{eq_lasso} is replaced by the $\ell_2$ norm), extracted all the seven looks containing double scatterers, and discarded one look whose height distance is almost identical to the one of another look. These six looks are shown in Fig.~\ref{fig_mean_ampl_double_scatterer}, where the indices increase with decreasing estimated height distance from approximately $1.5$ to $0.8$ times the vertical Rayleigh resolution. This ordering agrees approximately with intuition under the assumption that the roof is entirely flat: the higher the scatterer is on the facade, the less is its height distance to the roof.

\begin{figure}[!tp]
\centering
\includegraphics[width=0.4\textwidth]{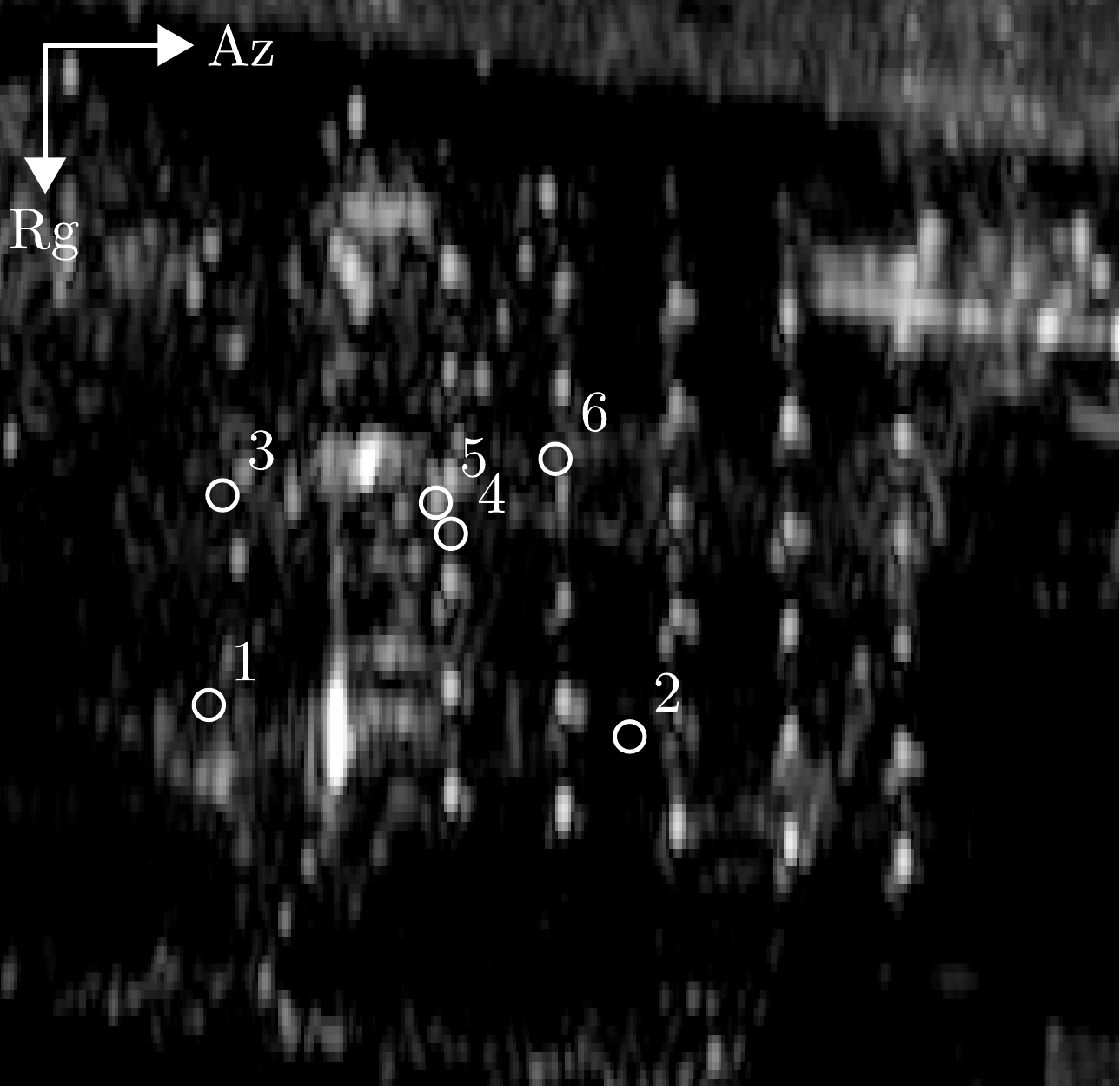}
\caption{Locations of six looks subject to facade-roof layover.} \label{fig_mean_ampl_double_scatterer}
\end{figure}

The estimated height profile is shown in Fig.~\ref{fig_x_h_sm_mm_1-3} (\#1--3) and \ref{fig_x_h_sm_mm_4-6} (\#4--6), where we used the vertical Rayleigh resolution of the \emph{single-master} stack (see Sec.~\ref{sec_experiment_design}) for normalizing the x-axis. The height estimates are listed in Tab.~\ref{tab_h_sm_mm}. In the \emph{single-master} setting, NLS and L1RLS produced very similar height profiles, despite the occasional sporadic artifacts in the latter which are known to be an intrinsic problem of $\ell_1$-regularization. Moreover, the height estimates were identical after off-grid correction. In each case, the height estimate of the lower scatterer fits very well the Level~2 RaySAR simulation of facade. Overall, the \emph{multi-master} results are consistent with the \emph{single-master} ones, with deviations of height estimates typically of several decimeters. In the \emph{fake single-master} setting, however, layover separation was only successful in the fifth case, presumably due to the high SNR (see the brightness of the look in Fig.~\ref{fig_mean_ampl_double_scatterer}). When the height distance is significantly larger than the vertical Rayleigh resolution (\#1--2), both NLS and L1RLS could reconstruct double scatterers, but only the one with larger amplitude could pass model-order selection. When the height distance approaches the vertical Rayleigh resolution or becomes even smaller (\#3, 4, 6), neither algorithm could reconstruct a second scatterer, and the height estimate of the single scatterer after off-grid correction is also arguably wrong. We are therefore convinced by this simple experiment that the conventional \emph{single-master} approach, if applied to a \emph{multi-master} stack, can be insufficient for layover separation.

\begin{figure*}[!tp]
\centering
\includegraphics[width=0.96\textwidth]{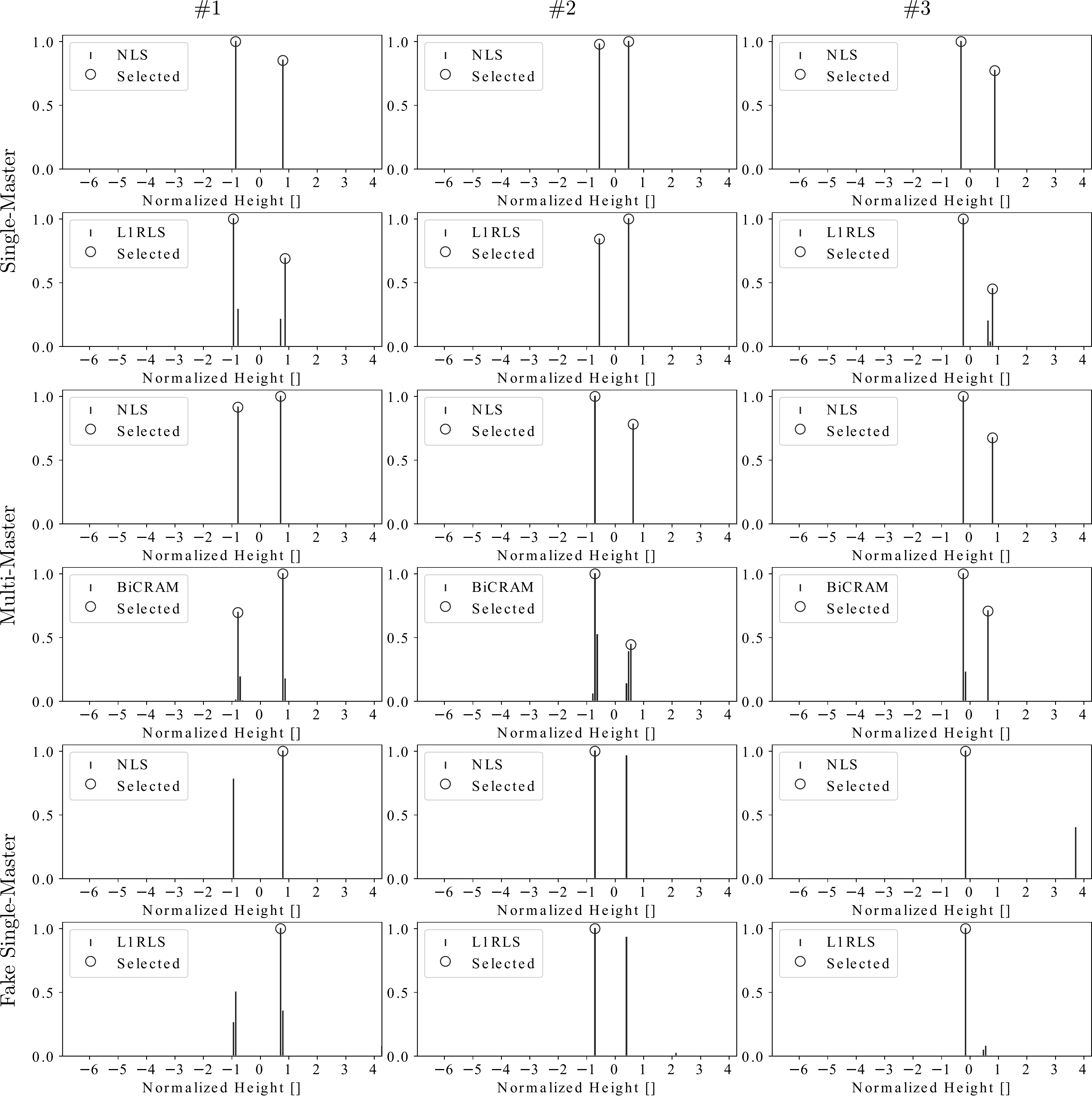}
\caption{Height profile estimate of \#1--3 in Fig.~\ref{fig_mean_ampl_double_scatterer}. Vertical line: before model-order selection. Circle: after model-order selection.} \label{fig_x_h_sm_mm_1-3}
\end{figure*}

\begin{figure*}[!tp]
\centering
\includegraphics[width=0.96\textwidth]{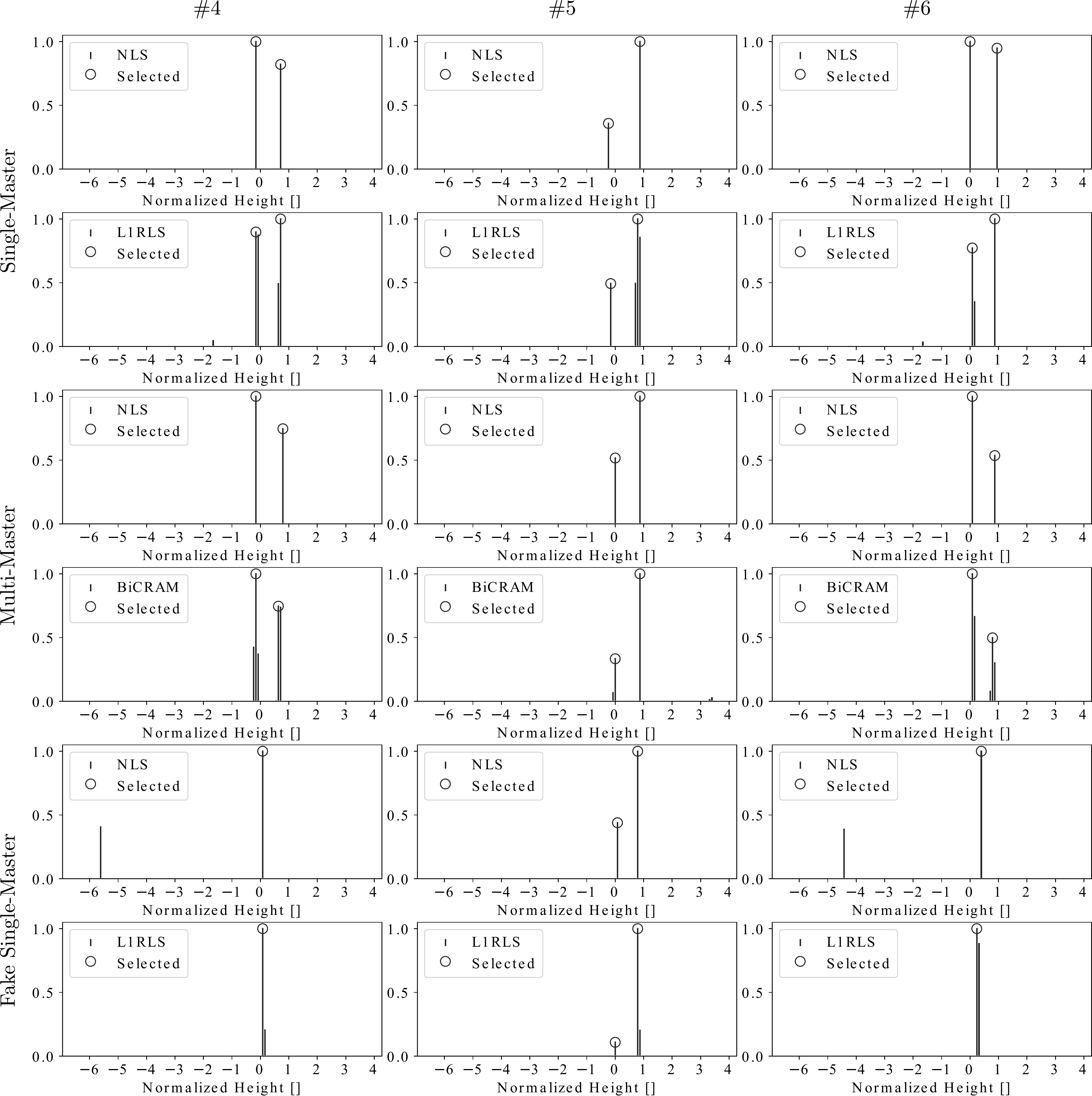}
\caption{Height profile estimate of \#4--6 in Fig.~\ref{fig_mean_ampl_double_scatterer}. Vertical line: before model-order selection. Circle: after model-order selection.} \label{fig_x_h_sm_mm_4-6}
\end{figure*}

\begin{table*}[!t]
    \renewcommand{\arraystretch}{1.3}
    \caption{Single- and multi-master height estimates of six layover cases [m]}
    \label{tab_h_sm_mm}
    \centering
    \begin{tabular}{ll|rr|rr|rr|rr|rr|rr}
        & & \multicolumn{2}{c|}{\#1} & \multicolumn{2}{c|}{\#2} & \multicolumn{2}{c|}{\#3} & \multicolumn{2}{c|}{\#4} & \multicolumn{2}{c|}{\#5} & \multicolumn{2}{c}{\#6} \\
        & & $h_1$ & $h_2$ & $h_1$ & $h_2$ & $h_1$ & $h_2$ & $h_1$ & $h_2$ & $h_1$ & $h_2$ & $h_1$ & $h_2$ \\
        \hline
        \multicolumn{2}{c|}{RaySAR} & $-9.22$ & $-$ & $-8.91$ & $-$ & $-2.39$ & $-$ & $-2.87$ & $-$ & $-1.59$ & $-$ & $0.55$ & $-$ \\
        \hline
        \multirow{2}{*}{Single-Master}
        & NLS & $-10.27$ & $9.10$ & $-8.14$ & $6.97$ & $-4.27$ & $10.72$ & $-2.63$ & $9.97$ & $-1.95$ & $10.03$ & $0.98$ & $11.34$ \\
        & L1RLS & $-10.27$ & $9.10$ & $-8.14$ & $6.97$ & $-4.27$ & $10.72$ & $-2.63$ & $9.97$ & $-1.95$ & $10.03$ & $0.98$ & $11.34$ \\
        \hline
        \multirow{2}{*}{Multi-Master}
        & NLS & $-10.15$ & $9.47$ & $-8.96$ & $8.23$ & $-3.50$ & $9.98$ & $-2.06$ & $10.45$ & $-0.37$ & $10.90$ & $0.66$ & $10.85$ \\
        & BiCRAM & $-10.15$ & $9.47$ & $-8.96$ & $8.23$ & $-3.50$ & $9.98$ & $-2.06$ & $10.45$ & $-0.37$ & $10.90$ & $0.66$ & $10.85$ \\
        \hline
        \multirow{2}{*}{Fake Single-Master}
        & NLS & $9.35$ & $-$ & $-9.12$ & $-$ & $-1.64$ & $-$ & $1.19$ & $-$ & $0.63$ & $10.03$ & $3.47$ & $-$ \\
        & L1RLS & $9.35$ & $-$ & $-9.12$ & $-$ & $-1.64$ & $-$ & $1.19$ & $-$ & $0.63$ & $10.03$ & $3.47$ & $-$ \\
        \hline
    \end{tabular}
\end{table*}



Naturally, we also performed tomographic reconstruction for all the $594$ looks within the building footprint in Fig.~\ref{fig_ps_x_y_sim_h} (bottom). Tab.~\ref{tab_runtime} lists the overall runtime on a desktop with a quad-core Intel processor at $3.40$~GHz and $16$-GB RAM. Note that the periodical deformation model was only used in the \emph{single-master} case, and the solution path of L1RLS or BiCRAM was sampled $11$ times (see Sec.~\ref{sec_experiment_design}). The height estimates of single and double scatterers are shown in Fig.~\ref{fig_ps_h_sm}--\ref{fig_ps_h_mm_wrong} for the three categories, respectively. In the case of double scatterers, the higher one was plotted. The seemingly messy appearance in the left column is due to the fact that single scatterers originate from both facade and roof. In spite of this, the gradual color transition at the $30$ PSs from far- to near-range agrees visually very well with the Level~1 ground truth in Fig.~\ref{fig_ps_x_y_sim_h} (top). Tab.~\ref{tab_no_scatterer_sm_mm} lists the number of scatterers in each case. In the \emph{single-scatterer} setting, NLS detected almost twice as many double scatterers as L1RLS. This is presumably due to a higher false positive rate: at $2$ out of $30$ PSs (fifth/second row from near range, and fifth/fifth column from late azimuth on the $6\times5$ regular grid) double scatterers were detected, although there should only be single ones. The number of double scatterers in the \emph{multi-master} case is in the same order as \emph{single-master} L1RLS, and the ratio between the number of single and the one of double scatterers is also similar. We attribute the smaller number of single scatterers to the nonconvexity of the optimization problem. In particular, as Thm.~\ref{thm_general_case}\eqref{thm_general_case_cond} suggests, a certain condition needs to be fulfilled for any nonzero solution of height profile estimate to exist at all, let alone whether an algorithm can provably recover it. In the \emph{fake single-master} category, many fewer double scatterers were produced. This is presumably due to double scatterers being misdetected as single scatterers, which occurred $5$ out of $6$ times in the previous experiment (see Fig.~\ref{fig_x_h_sm_mm_1-3} and \ref{fig_x_h_sm_mm_4-6}).

The next subsection elucidates how we validated the height estimates with Level~1 and 2 ground truth data.

\begin{table}[!tbp]
    \renewcommand{\arraystretch}{1.3}
    \caption{Single- and multi-master runtime}
    \label{tab_runtime}
    \centering
    \begin{tabular}{ll|r}
        & & Runtime [s] \\
        \hline
        \multirow{2}{*}{Single-Master}
        & NLS & $6154$ \\
        & L1RLS & $736$ \\
        \hline
        \multirow{2}{*}{Multi-Master}
        & NLS & $460$ \\
        & BiCRAM & $6853$ \\
        \hline
        \multirow{2}{*}{Fake Single-Master}
        & NLS & $48$ \\
        & L1RLS & $65$ \\
        \hline
    \end{tabular}
\end{table}

\begin{figure*}[!tp]
\centering
\includegraphics[width=0.6\textwidth]{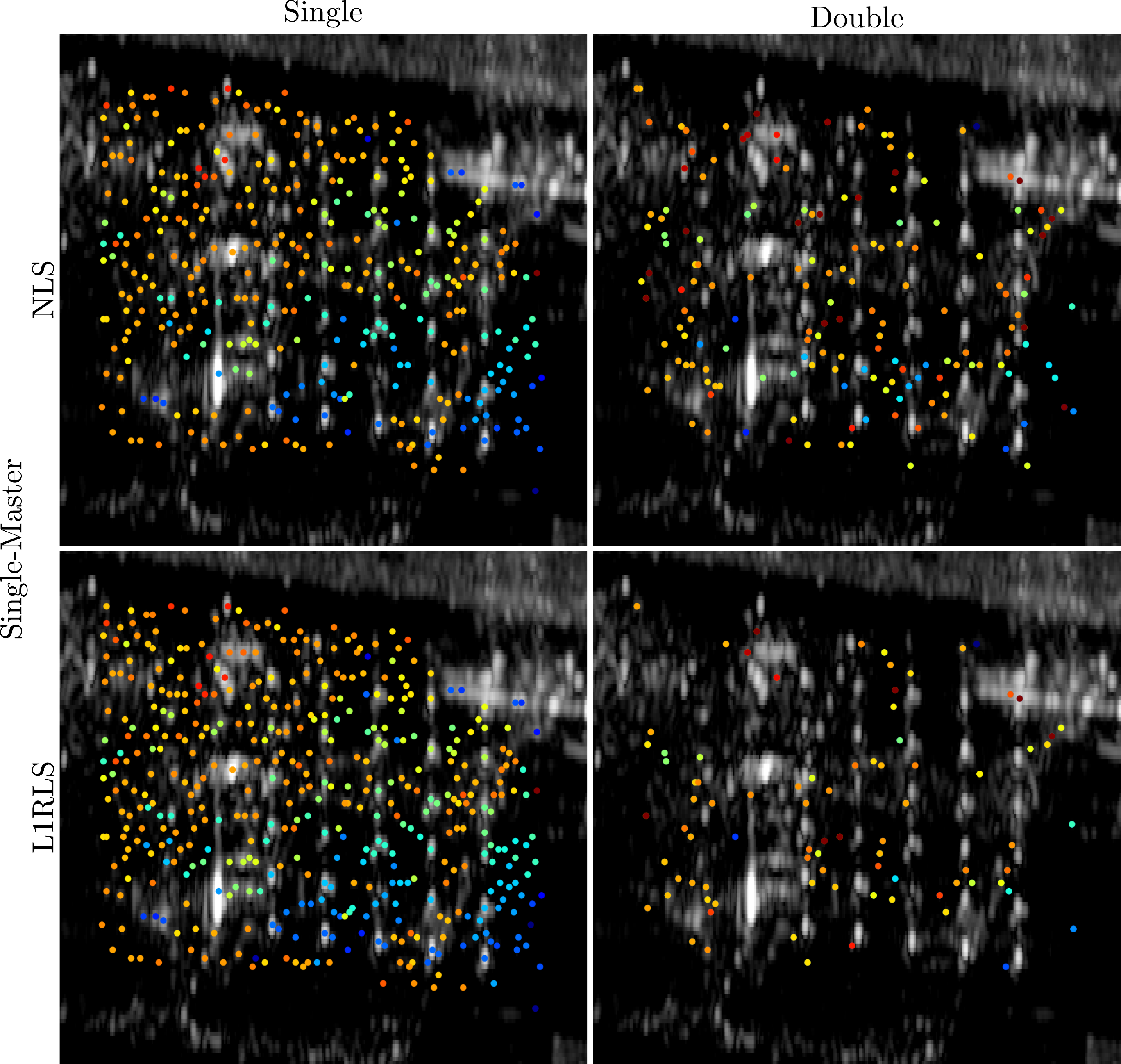} \\
\includegraphics[width=0.35\textwidth]{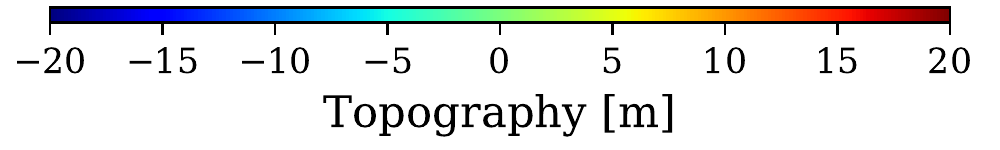}
\caption{Single-master height estimates of single and double scatterers. Top: NLS. Bottom: L1RLS.} \label{fig_ps_h_sm}
\end{figure*}

\begin{figure*}[!tp]
\centering
\includegraphics[width=0.6\textwidth]{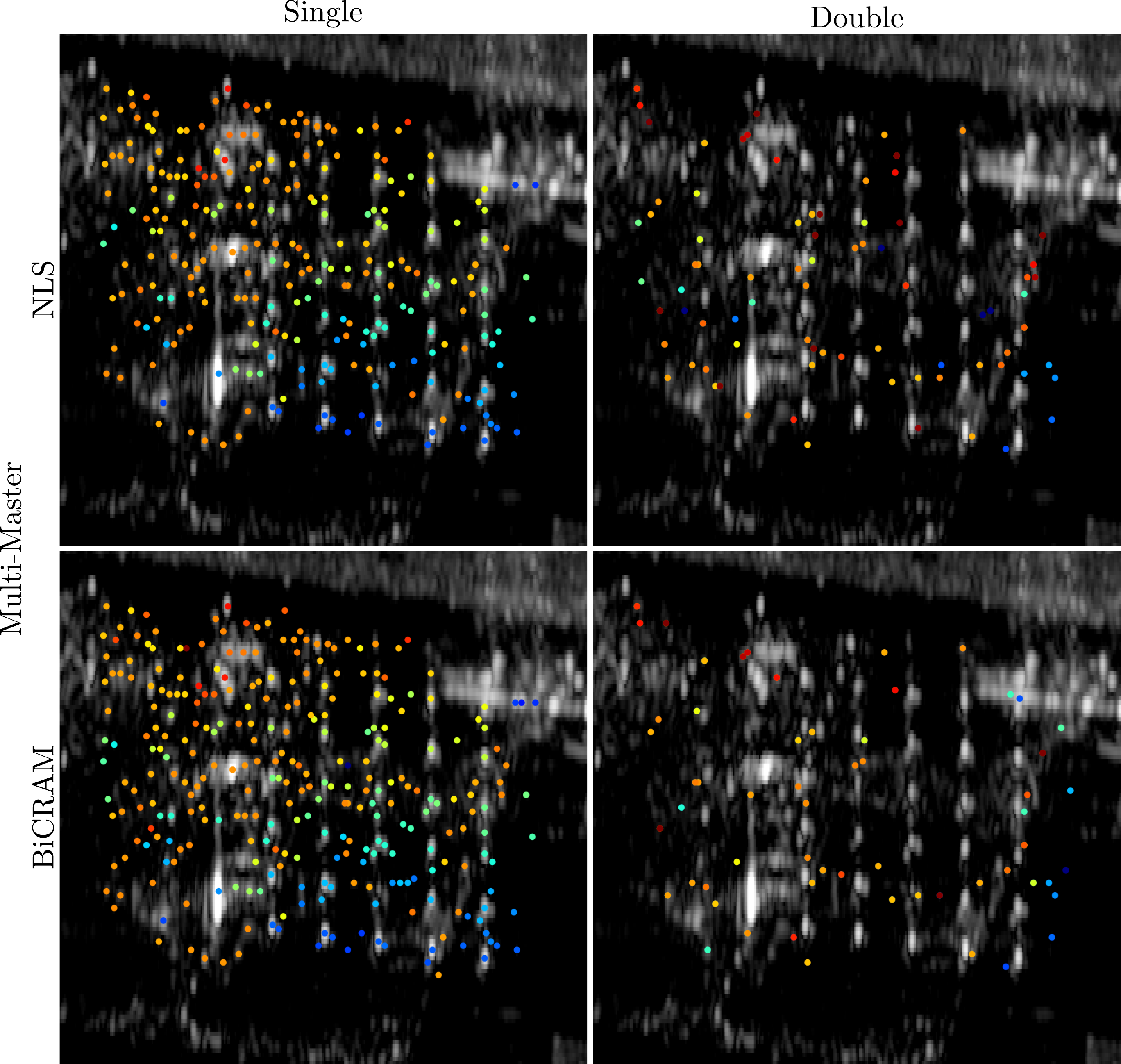} \\
\includegraphics[width=0.35\textwidth]{fig/fig_cbar_h_-20_20}
\caption{Multi-master height estimates of single and double scatterers. Top: NLS. Bottom: BiCRAM.} \label{fig_ps_h_mm}
\end{figure*}

\begin{figure*}[!tp]
\centering
\includegraphics[width=0.6\textwidth]{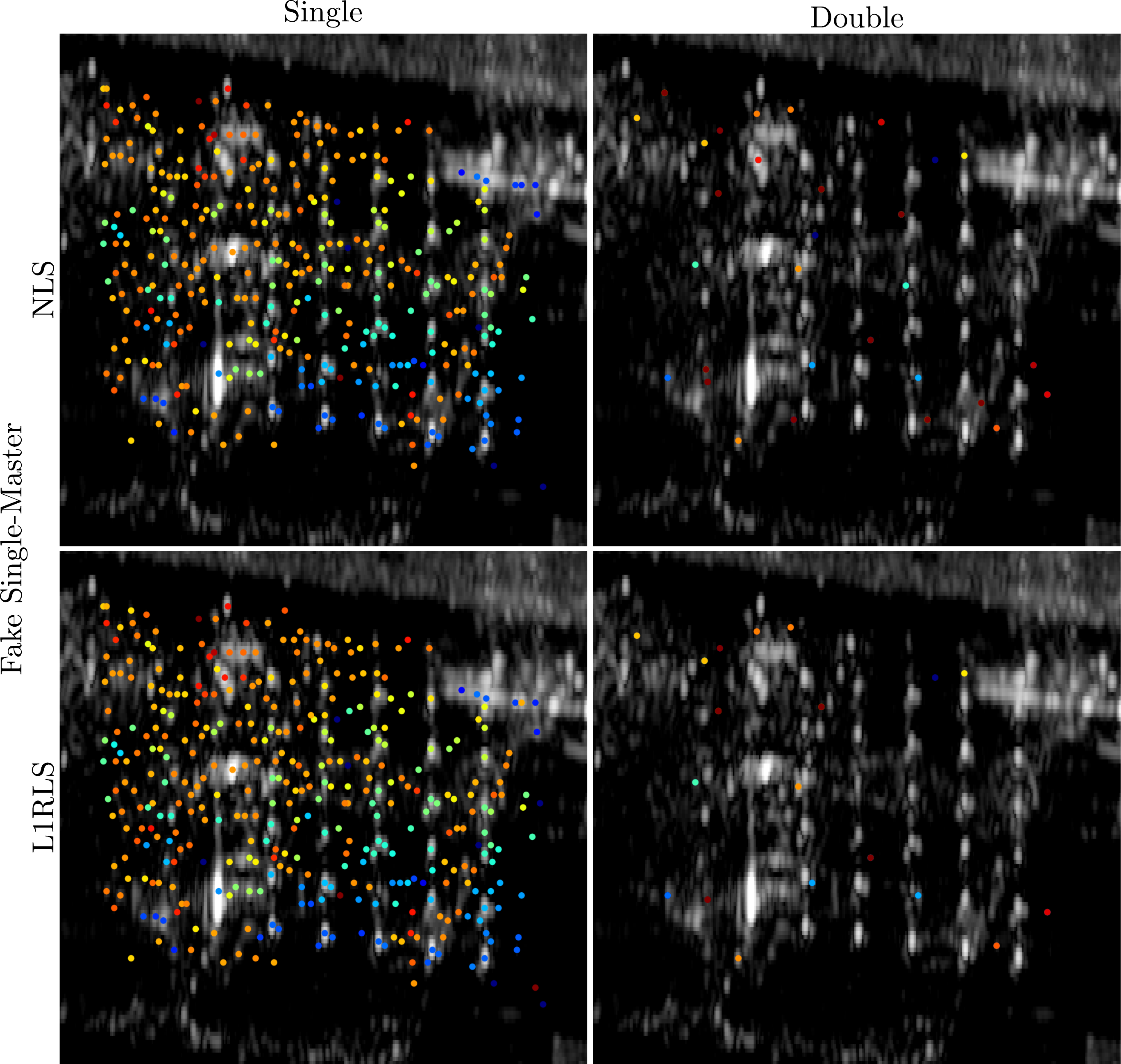} \\
\includegraphics[width=0.35\textwidth]{fig/fig_cbar_h_-20_20}
\caption{Fake single-master height estimates of single and double scatterers. Top: NLS. Bottom: L1RLS.} \label{fig_ps_h_mm_wrong}
\end{figure*}

\begin{table}[!tbp]
    \renewcommand{\arraystretch}{1.3}
    \caption{Single- and multi-master number of scatterers}
    \label{tab_no_scatterer_sm_mm}
    \centering
    \begin{tabular}{ll|r|r|r|r}
        & & Single & Double & Ratio & Facade \\
        \hline
        \multirow{2}{*}{Single-Master}
        & NLS & $359$ & $332$ & $1.08$ & $148$ \\
        & L1RLS & $446$ & $168$ & $2.65$ & $189$ \\
        \hline
        \multirow{2}{*}{Multi-Master}
        & NLS & $260$ & $158$ & $1.65$ & $124$ \\
        & BiCRAM & $291$ & $118$ & $2.47$ & $133$ \\
        \hline
        \multirow{2}{*}{Fake Single-Master}
        & NLS & $360$ & $60$ & $6.00$ & $134$ \\
        & L1RLS & $381$ & $38$ & $10.03$ & $143$ \\
        \hline
    \end{tabular}
\end{table}


\subsection{Validation} \label{sec_experiment_validation}

Since the height ground truth is limited to facade only (see Sec.~\ref{sec_experiment_generation}), the validation was focused on single scatterers by following two approaches: the first one uses $30$ PSs, and the second one is based on extracted facade scatterers.

As already mentioned in Sec.~\ref{sec_experiment_design}, the $30$ PSs constituting the Level~1 ground truth in Fig.~\ref{fig_ps_x_y_sim_h} (top) are caused by triple reflections on the building facade, and are located on a regular grid of salient points. Due to the (almost) identical scattering geometry, these PSs should have similar SNRs and are therefore ideal for height estimate validation. In each of the six cases, single scatterers were correctly detected at all the $30$ PSs---with the exception that double scatterers were misdetected by NLS in the \emph{single-master} setting (see Sec.~\ref{sec_experiment_experimental}). For this reason, the height estimate error could be evaluated straightforwardly. Fig.~\ref{fig_ps_topo_err_sm_mm_ps} shows the normalized histogram and Tab.~\ref{tab_height_err_ps} lists some of its statistical parameters. As a reference, the SD and MAD of height estimate error of the PSI result are about $0.28$ and $0.22$~m, respectively \cite{auer2019sigs}. In each of the three settings (single-master, multi-master and fake single-master), the respective two algorithms performed similarly and no significant difference is visible. A cross-comparison between the multi-master and fake single-master cases revealed the superiority of the former: its histogram is more centered around zero, and both its SD and MAD are slightly smaller. This is unsurprising since we already analyzed the implications of the single-look \emph{multi-master} data model for single scatterers in Sec.~\ref{sec_multimaster_implications}. At this point, we could confidently assert that it does make a difference in practice, albeit small, despite the longer (by approximately one order considering that the solution path of BiCRAM was sampled $11$ times) processing time. Somewhat surprisingly, the \emph{multi-master} result is also slightly better than the \emph{single-master} one. We suspect that this is due to the complication of \emph{single-master} tomographic processing by using the (imperfect) periodical deformation model, and the justified simplification in the \emph{multi-master} case thanks to the small-baseline configuration (so that deformation-induced phase is mitigated via forming interferograms).

The second approach is based on all facade scatterers (Level~2 ground truth) in Fig.~\ref{fig_ps_x_y_sim_h} (bottom), given that they do exist. Scatter plots of simulated and estimated height of single scatterers are shown in Fig.~\ref{fig_ps_sim_topo_scatterplot}. It is obvious that many single scatterers are located on the building roof (see the gray dots above diagonal line). In order to extract facade scatterers, we used a threshold of $\pm 3 \times \text{CRLB}$ added to the simulated value. The extracted facade scatterers, whose number is given in Tab.~\ref{tab_no_scatterer_sm_mm} for each case, are shown as black dots and were used for height estimate validation. The normalized histogram is shown in Fig.~\ref{fig_ps_topo_err_sm_mm_tomosar} and some of its statistical parameters are provided in Tab.~\ref{tab_height_err_tomosar}. Likewise, the two respective algorithms performed similarly in each setting, and the \emph{multi-master} height estimate error has slightly less deviation. The SD and MAD are worse in comparison to those in Tab.~\ref{tab_height_err_tomosar} due to the much larger range of SNRs.

\begin{figure}[!tp]
\centering
\includegraphics[width=0.4\textwidth]{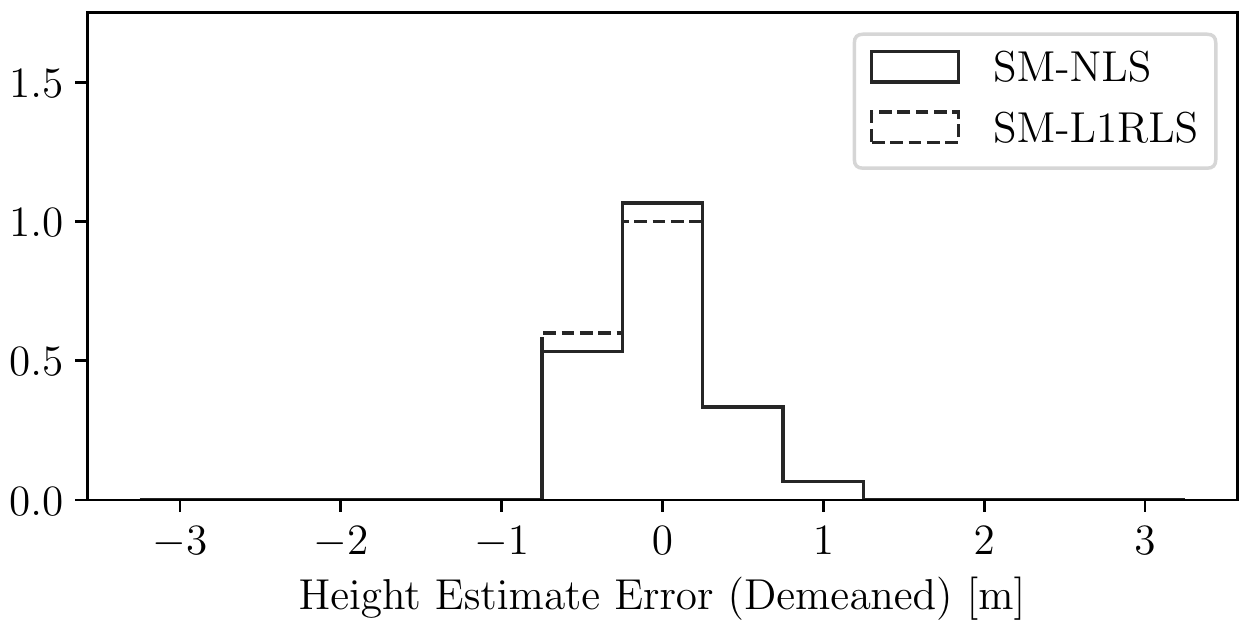} \\
\includegraphics[width=0.4\textwidth]{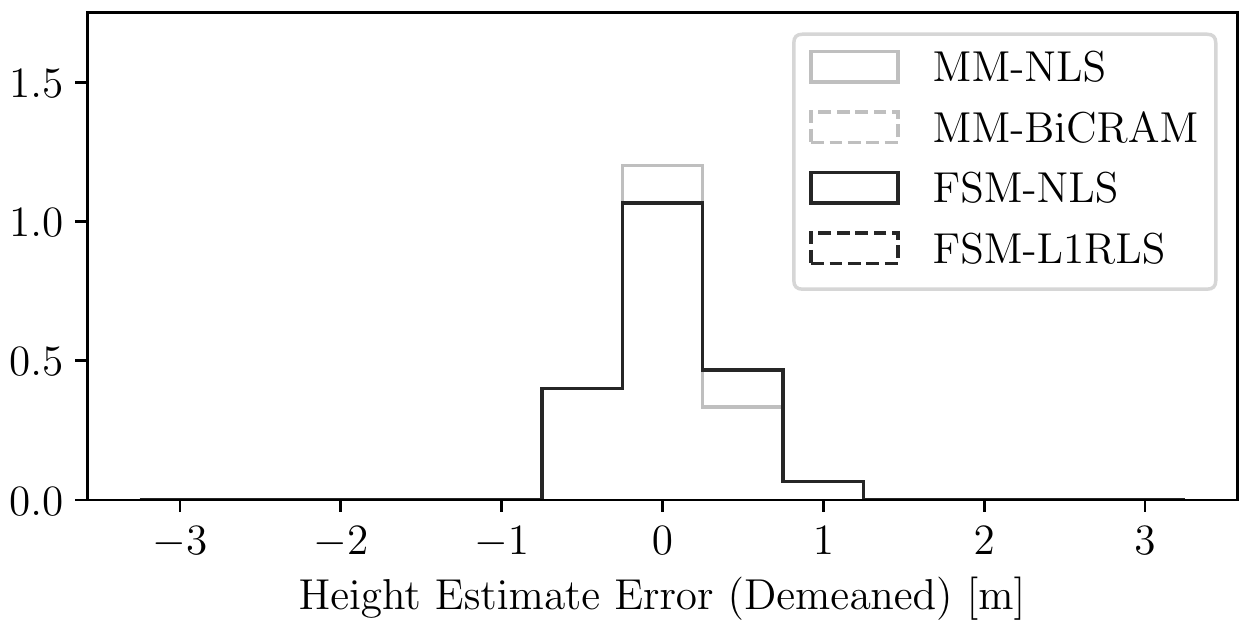}
\caption{Normalized histogram of height estimate error of $30$ PSs (Level~1). SM: single-master. MM: multi-master. FSM: fake single-master.} \label{fig_ps_topo_err_sm_mm_ps}
\end{figure}

\begin{table*}[!tbp]
    \renewcommand{\arraystretch}{1.3}
    \caption{Statistics of height estimate error [m]: $30$ PSs (Level~1)}
    \label{tab_height_err_ps}
    \centering
    \begin{tabular}{ll|r|r|r|r|r|r}
        & & Min & Max & Mean & Median & SD & MAD \\
        \hline
        \multirow{2}{*}{Single-Master}
        & NLS & $-0.81$ & $0.65$ & $-0.30$ & $-0.29$ & $0.34$ & $0.35$ \\
        & L1RLS & $-0.81$ & $0.65$ & $-0.31$ & $-0.29$ & $0.34$ & $0.37$ \\
        \hline
        \multirow{2}{*}{Multi-Master}
        & NLS & $-0.88$ & $0.55$ & $-0.33$ & $-0.36$ & $0.31$ & $0.28$ \\
        & BiCRAM & $-0.88$ & $0.55$ & $-0.33$ & $-0.36$ & $0.31$ & $0.28$ \\
        \hline
        \multirow{2}{*}{Fake Single-Master}
        & NLS & $-1.02$ & $0.68$ & $-0.43$ & $-0.48$ & $0.34$ & $0.30$ \\
        & L1RLS & $-1.02$ & $0.68$ & $-0.43$ & $-0.48$ & $0.34$ & $0.30$ \\
        \hline
    \end{tabular}
\end{table*}

\begin{figure}[!tp]
\centering
\includegraphics[width=0.4\textwidth]{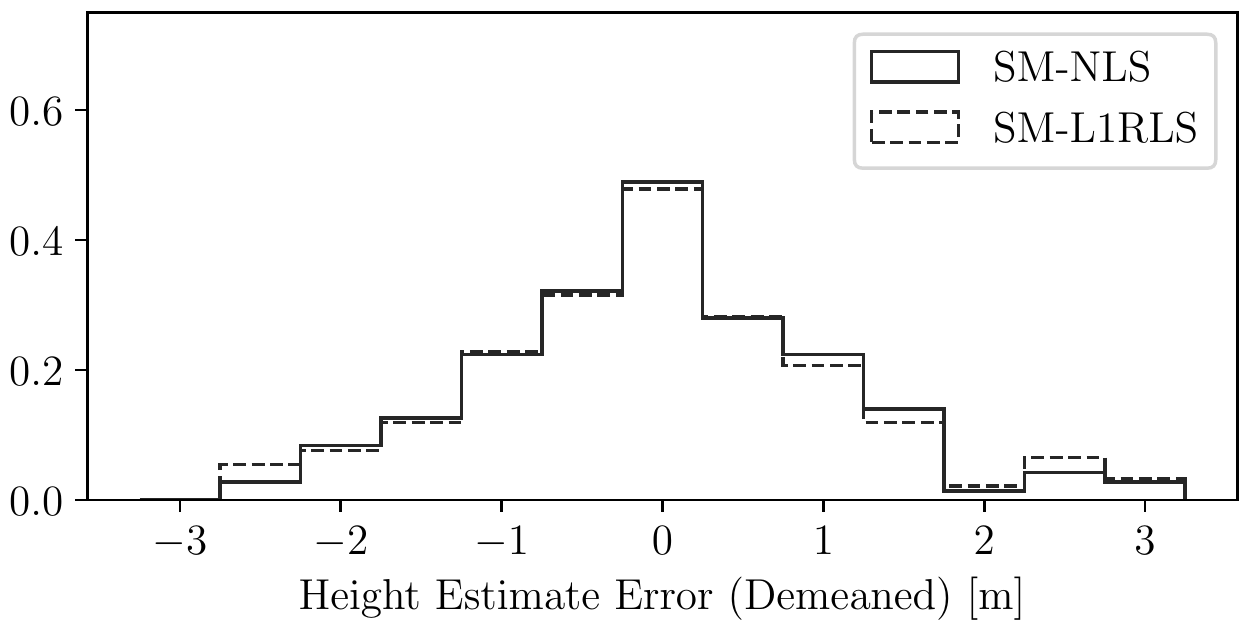} \\
\includegraphics[width=0.4\textwidth]{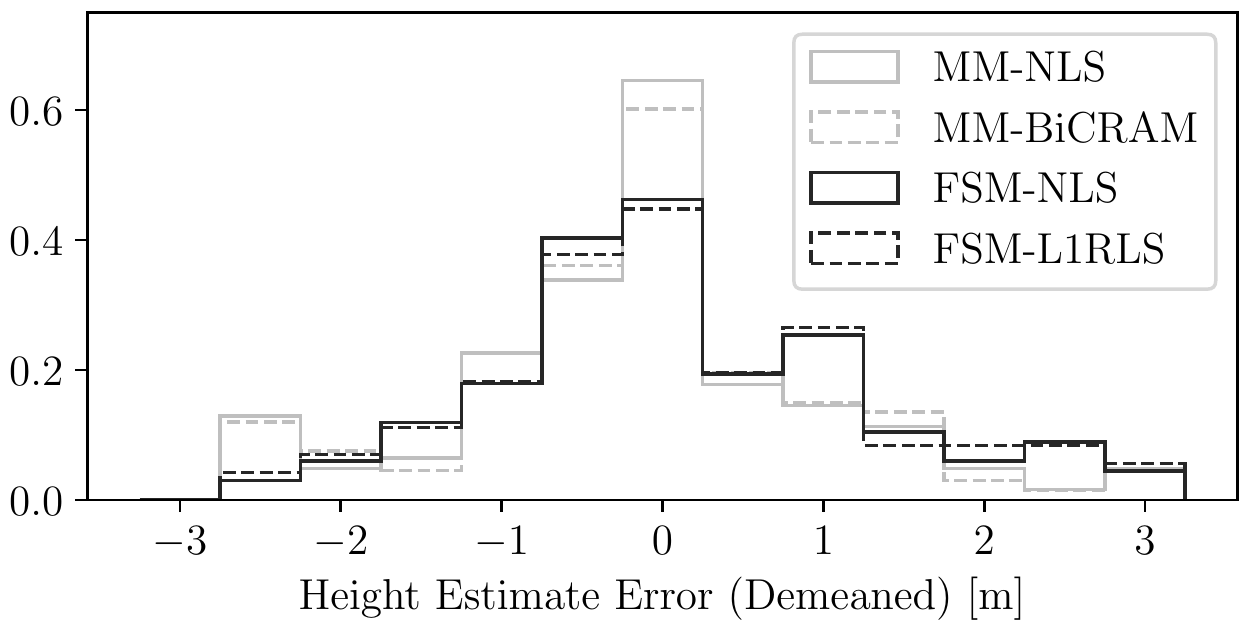}
\caption{Normalized histogram of height estimate error of extracted facade scatterers (Level~2). SM: single-master. MM: multi-master. FSM: fake single-master.} \label{fig_ps_topo_err_sm_mm_tomosar}
\end{figure}

\begin{figure*}[!tp]
\centering
\includegraphics[width=0.6\textwidth]{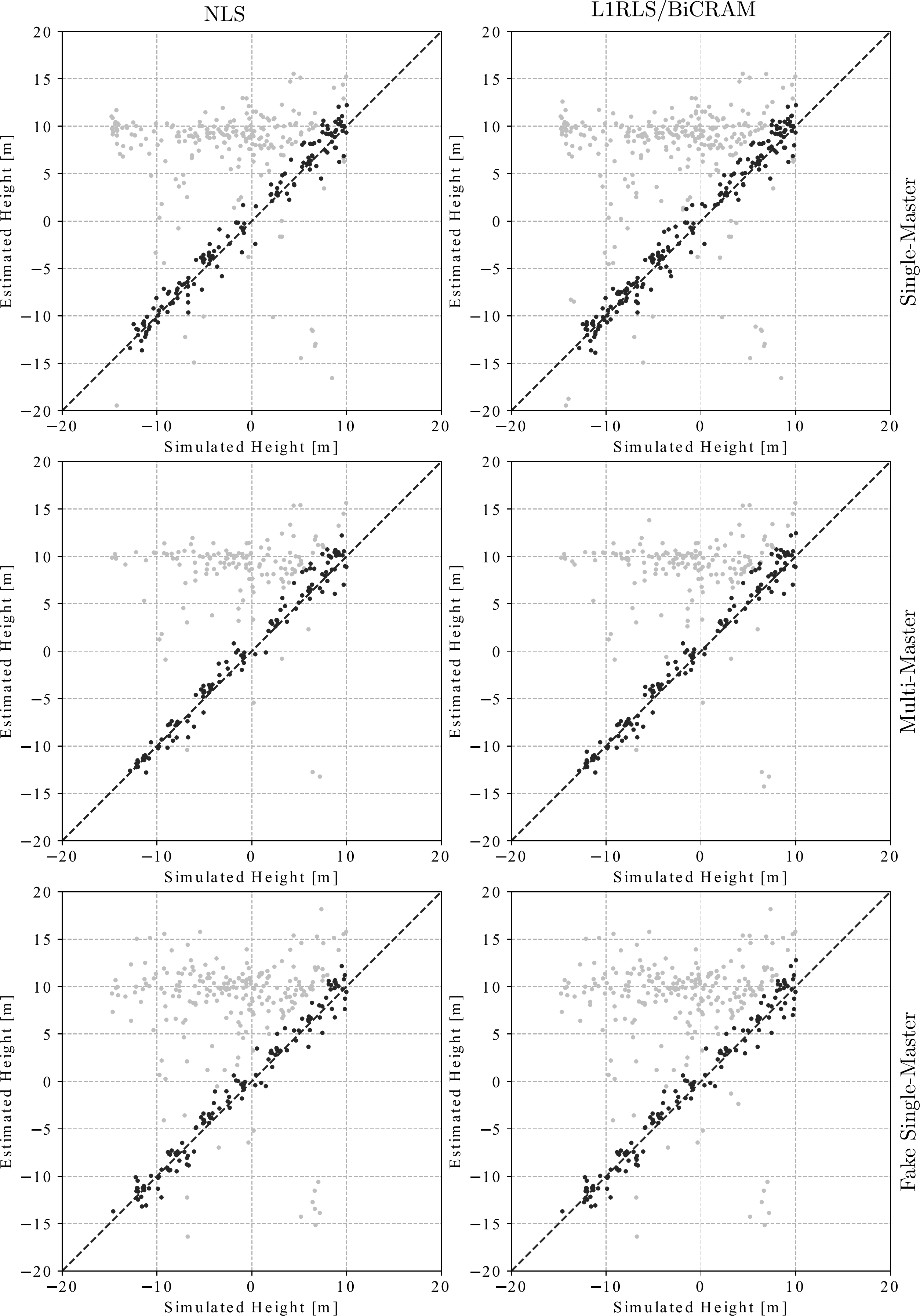}
\caption{Scatter plot of simulated and estimated height of single scatterers using Level~2 height ground truth. Black: extracted facade scatterers. Gray: extracted non-facade scatterers.} \label{fig_ps_sim_topo_scatterplot}
\end{figure*}

\begin{table*}[!tbp]
    \renewcommand{\arraystretch}{1.3}
    \caption{Statistics of height estimate error [m]: extracted facade scatterers (Level~2)}
    \label{tab_height_err_tomosar}
    \centering
    \begin{tabular}{ll|r|r|r|r|r|r}
        & & Min & Max & Mean & Median & SD & MAD \\
        \hline
        \multirow{2}{*}{Single-Master}
        & NLS & $-3.03$ & $2.99$ & $-0.40$ & $-0.51$ & $1.23$ & $0.96$ \\
        & L1RLS & $-3.03$ & $2.99$ & $-0.42$ & $-0.49$ & $1.25$ & $0.99$ \\
        \hline
        \multirow{2}{*}{Multi-Master}
        & NLS & $-3.04$ & $2.71$ & $-0.49$ & $-0.40$ & $1.13$ & $0.85$ \\
        & BiCRAM & $-3.04$ & $2.71$ & $-0.50$ & $-0.40$ & $1.12$ & $0.86$ \\
        \hline
        \multirow{2}{*}{Fake Single-Master}
        & NLS & $-2.93$ & $2.70$ & $-0.38$ & $-0.50$ & $1.15$ & $0.91$ \\
        & L1RLS & $-2.93$ & $2.70$ & $-0.35$ & $-0.48$ & $1.21$ & $1.05$ \\
        \hline
    \end{tabular}
\end{table*}

The next section concludes this paper and suggests some prospective work.

\section{Conclusion and discussion} \label{sec_conclusion}

The previous sections provided new insights into single-look multi-master SAR tomography. The single-look multi-master data model was established and two algorithms were developed within a common inversion framework. The first algorithm extends the conventional NLS to the single-look multi-master data model, and the second one uses bi-convex relaxation and alternating minimization. Extensive efforts were devoted to studying the nonconvex objective function of the NLS subproblem, and to experimenting with different acceleration techniques for ADMM-based algorithms. We demonstrated with the help of a real TerraSAR-X data set that the conventional \emph{single-master} approach, if applied to a \emph{multi-master} stack, can be insufficient for layover separation, even when the height distance between two scatterers is significantly larger than the vertical Rayleigh resolution. By means of a SAR imaging geodesy and simulation framework, we managed to generate two levels of height ground truth. The height estimates in each of the three settings were validated at either $30$ PSs or hundreds of extracted facade scatterers. Overall, the \emph{multi-master} approach performed slightly better, although it was computationally more demanding.

A special case of the general problem analyzed so far is single-look SAR tomography using only bistatic (or pursuit monostatic) interferograms. On the one hand, the advantages are that bistatic interferograms are (almost) APS-free, and the data model is still linear for any single scatterer whose reflectivity can be estimated up to a constant phase angle \eqref{eq_mm_single}. On the other, the disadvantages are that, for double or multiple scatterers, bistatic interferograms are not motion-free (see Sec.~\ref{sec_multimaster_implications}), and the data model is nonlinear \eqref{eq_mm_double}.

An alternative way to formulate the problem in the bistatic setting is to parameterize APS without forming any interferogram. Let $\bfg$ and $\bfh$ be the bistatic observations of the master and slave scenes, respectively. We have essentially two data sets:
\begin{equation}
    \bfg \approx (\bfR\bfgamma) \had \exp(j\bfphi), \quad \bfh \approx (\bfS\bfgamma) \had \exp(j\bfphi),
\end{equation}
where $\bfphi\in\real^{N'}$ is the APS vector. Under the sparsity or compressibility assumption of $\bfgamma$, one could consider the following problem:
\begin{equation} \label{eq_bistatic}
\begin{aligned}
    & \minimize_{\bfgamma,\bfphi} && \frac{1}{2} \|(\bfR\bfgamma) \had \exp(j\bfphi) - \bfg\|_2^2 + \\
    & && \frac{1}{2} \|(\bfS\bfgamma) \had \exp(j\bfphi) - \bfh\|_2^2 + \lambda \|\bfgamma\|_1,
\end{aligned}
\end{equation}
or more compactly
\begin{equation}
    \minimize_{\bfgamma,\bfphi} \frac{1}{2} \|(\tilde\bfR\bfgamma) \had (\tilde\bfI\exp(j\bfphi)) - \tilde\bfg\|_2^2 + \lambda \|\bfgamma\|_1,
\end{equation}
where $\tilde\bfR \isdef \begin{pmatrix}
\bfR \\
\bfS
\end{pmatrix}$, $\tilde\bfI \isdef \begin{pmatrix}
\bfI \\
\bfI
\end{pmatrix}$, and $\tilde\bfg \isdef \begin{pmatrix}
\bfg \\
\bfh
\end{pmatrix}$. Note that this problem is bi-convex in $\bfgamma$ and $\bfphi$.

Inspired by PSI, another related problem is SAR tomography on edges. Let $\bfgamma$ and $\bftheta$ represent the reflectivity profiles of two neighboring looks, and their phase-calibrated SLC measurements be denoted as $\bfg$ and $\bfh$, respectively. Consider the following problem:
\begin{equation}
    \minimize_{\bfgamma,\bftheta} \frac{1}{2} \|(\bfR\bfgamma) \had (\conj{\bfS\bftheta}) - \bfg \had \conj{\bfh}\|_2^2 + \lambda_1 \|\bfgamma\|_1 + \lambda_2 \|\bftheta\|_1,
\end{equation}
which is bi-convex in $\bfgamma$ and $\bftheta$. Likewise, the rationale of $\bfg \had \conj{\bfh}$ is to mitigate APS for neighboring looks. Alternatively, a parametric approach similar to \eqref{eq_bistatic} could also be considered.


%




\appendices

\section{Recap of ADMM} \label{app_admm}

ADMM \cite{boyd2011distributed} solves a minimization problem in the form of
\begin{equation}
\begin{aligned}
    & \minimize_{\bfx,\bfz} && f(\bfx) + g(\bfz) \\
    & \st && \bfC\bfx + \bfD\bfz = \bfe
\end{aligned}
\end{equation}
by alternatively minimizing its augmented Lagrangian \cite[p.~509]{foucart2017mathematical}
\begin{equation} \label{eq_aug_lag}
\begin{aligned}
    L_\rho(\bfx,\bfy,\bfz) \isdef {}& f(\bfx) + g(\bfz) + \Real\langle\bfy, \bfC\bfx+\bfD\bfz-\bfe\rangle + \\
    & (\rho/2) \|\bfC\bfx+\bfD\bfz-\bfe\|_2^2,
\end{aligned}
\end{equation}
i.e.,
\begin{equation} \label{eq_admm}
\begin{aligned}
    \bfx^{(k+1)} \isdef {}& \argmin_{\bfx} L_\rho(\bfx,\bfy^{(k)},\bfz^{(k)}) \\
    \bfz^{(k+1)} \isdef {}& \argmin_{\bfz} L_\rho(\bfx^{(k+1)},\bfy^{(k)},\bfz) \\
    \bfy^{(k+1)} \isdef {}& \bfy^{(k)} + \rho(\bfC\bfx^{(k+1)}+\bfD\bfz^{(k+1)}-\bfe)
\end{aligned}
\end{equation}
in the $k$th iteration, where $\rho\in\real_{++}$ is a penalty parameter.

\section{Proof of Proposition~\ref{thm_eig_hess}} \label{app_prop}

The proof uses the following minor result.
\begin{lemma} \label{thm_block_mat_inv}
For any $\bfF,\bfG\in\real^{n \times n}$ and $c,d\in\real$ such that $c^2+d^2=1$, the following equalities hold:
\begin{equation}
\begin{aligned}
    &
    \begin{pmatrix}
    c\bfI & d\bfI \\
    -d\bfI & c\bfI
    \end{pmatrix}^{-1}
    \begin{pmatrix}
    \bfF & -\bfG \\
    \bfG & \bfF
    \end{pmatrix}
    \begin{pmatrix}
    c\bfI & d\bfI \\
    -d\bfI & c\bfI
    \end{pmatrix}
    =
    \begin{pmatrix}
    \bfF & -\bfG \\
    \bfG & \bfF
    \end{pmatrix}, \\
    &
    \begin{pmatrix}
    c\bfI & d\bfI \\
    -d\bfI & c\bfI
    \end{pmatrix}^{-1}
    \begin{pmatrix}
    \bfF & \bfG \\
    \bfG & -\bfF
    \end{pmatrix}
    =
    \begin{pmatrix}
    \bfF & \bfG \\
    \bfG & -\bfF
    \end{pmatrix}
    \begin{pmatrix}
    c\bfI & d\bfI \\
    -d\bfI & c\bfI
    \end{pmatrix}.
\end{aligned}
\end{equation}
\end{lemma}
\begin{proof}
Observe that for any $a,b\in\real$ such that $a^2+b^2 \neq 0$,
\begin{equation}
    \begin{pmatrix}
    a\bfI & b\bfI \\
    -b\bfI & a\bfI
    \end{pmatrix}^{-1}
    = \frac{1}{a^2+b^2}
    \begin{pmatrix}
    a\bfI & -b\bfI \\
    b\bfI & a\bfI
    \end{pmatrix}.
\end{equation}
The rest of the proof follows via straightforward computations.
\end{proof}

Now we turn our attention to the proposition.
\begin{proof}[Proof of Proposition~\ref{thm_eig_hess}]
First, we prove that $\nabla^2f(\bfx)$ and $\nabla^2f \left( \bfx\exp(j\phi) \right)$ are similar, i.e., there exists an invertible $\bfP$ such that $\nabla^2f(\bfx) = \bfP^{-1} \nabla^2f \left( \bfx\exp(j\phi) \right) \bfP$.

Observe that
\begin{equation}
\begin{aligned}
    & && \bfC\left(\bfx\exp(j\phi)\right) \\
    &= && \bfA^H \Diag\left( (\bfB\bfx\exp(j\phi)) \had (\conj{\bfB\bfx\exp(j\phi)}) \right) \bfA +{} \\
    & && \bfB^H \Diag\left( (\bfA\bfx\exp(j\phi)) \had (\conj{\bfA\bfx\exp(j\phi)}) \right) \bfB \\
    &= && \bfC(\bfx), \\
    & && \bfD\left(\bfx\exp(j\phi)\right) \\
    &= && \bfA^H \Diag\left( (\bfA\bfx\exp(j\phi)) \had (\bfB\bfx\exp(j\phi)) \right) \conj{\bfB} +{} \\
    & && \bfB^H \Diag\left( (\bfA\bfx\exp(j\phi)) \had (\bfB\bfx\exp(j\phi)) \right) \conj{\bfA} \\
    &= && \bfD(\bfx) \exp(j2\phi), \\
    & && \bfE\left(\bfx\exp(j\phi)\right) \\
    &= && \bfA^H \Diag\left( (\bfA\bfx\exp(j\phi)) \had (\conj{\bfB\bfx\exp(j\phi)}) - \bfb \right) \bfB +{} \\
    & && \bfB^H \Diag\left( (\conj{\bfA\bfx\exp(j\phi)}) \had (\bfB\bfx\exp(j\phi)) - \conj{\bfb} \right) \bfA \\
    &= && \bfE(\bfx).
\end{aligned}
\end{equation}
Let $\bfC \isdef \bfC(\bfx)$, $\bfD \isdef \bfD(\bfx)$, $\bfE \isdef \bfE(\bfx)$. The Hessian becomes
\begin{equation}
\begin{aligned}
    & && \nabla^2f \left( \bfx\exp(j\phi) \right) \\
    &= &&
    \begin{pmatrix}
    \Real(\bfC) & -\Imag(\bfC) \\
    \Imag(\bfC) & \Real(\bfC) \\
    \end{pmatrix} +{} \\
    & &&
    \begin{pmatrix}
    \Real\left(\bfD \exp(j2\phi)\right) & \Imag\left(\bfD \exp(j2\phi)\right) \\
    \Imag\left(\bfD \exp(j2\phi)\right) & -\Real\left(\bfD \exp(j2\phi)\right)
    \end{pmatrix} +{} \\
    & &&
    \begin{pmatrix}
    \Real(\bfE) & -\Imag(\bfE) \\
    \Imag(\bfE) & \Real(\bfE) \\
    \end{pmatrix} \\
    &= &&
    \begin{pmatrix}
    \Real(\bfC+\bfE) & -\Imag(\bfC+\bfE) \\
    \Imag(\bfC+\bfE) & \Real(\bfC+\bfE) \\
    \end{pmatrix} +{} \\
    & &&
    \begin{pmatrix}
    \Real(\bfD) & \Imag(\bfD) \\
    \Imag(\bfD) & -\Real(\bfD) \\
    \end{pmatrix}
    \begin{pmatrix}
    \cos(2\phi)\bfI & \sin(2\phi)\bfI \\
    -\sin(2\phi)\bfI & \cos(2\phi)\bfI
    \end{pmatrix}.
\end{aligned}
\end{equation}

The choice of $\bfP$ can be divided into two cases depending on the value of $\phi$.

In the trivial case, $\phi = (2k+1)\pi/2$ for some $k \in \intg$. Let
\begin{equation}
    \bfP \isdef
    \begin{pmatrix}
    \bfzero & -\bfI \\
    \bfI & \bfzero
    \end{pmatrix}.
\end{equation}
This leads to
\begin{equation}
\begin{aligned}
    & && \bfP^{-1} \nabla^2f \left( \bfx\exp(j\phi) \right) \bfP \\
    &= &&
    \begin{pmatrix}
    \bfzero & -\bfI \\
    \bfI & \bfzero
    \end{pmatrix}^{-1}
    \begin{pmatrix}
    \Real(\bfC+\bfE) & -\Imag(\bfC+\bfE) \\
    \Imag(\bfC+\bfE) & \Real(\bfC+\bfE) \\
    \end{pmatrix}
    \begin{pmatrix}
    \bfzero & -\bfI \\
    \bfI & \bfzero
    \end{pmatrix} +{} \\
    & &&
    \begin{pmatrix}
    \bfzero & -\bfI \\
    \bfI & \bfzero
    \end{pmatrix}^{-1}
    \begin{pmatrix}
    \Real(\bfD) & \Imag(\bfD) \\
    \Imag(\bfD) & -\Real(\bfD) \\
    \end{pmatrix}
    \begin{pmatrix}
    -\bfI & \bfzero \\
    \bfzero & -\bfI
    \end{pmatrix}
    \begin{pmatrix}
    \bfzero & -\bfI \\
    \bfI & \bfzero
    \end{pmatrix} \\
    &= &&
    \begin{pmatrix}
    \Real(\bfC+\bfE) & -\Imag(\bfC+\bfE) \\
    \Imag(\bfC+\bfE) & \Real(\bfC+\bfE) \\
    \end{pmatrix} +{} \\
    & &&
    \begin{pmatrix}
    \Real(\bfD) & \Imag(\bfD) \\
    \Imag(\bfD) & -\Real(\bfD) \\
    \end{pmatrix}
    \begin{pmatrix}
    \bfzero & -\bfI \\
    \bfI & \bfzero
    \end{pmatrix}
    \begin{pmatrix}
    -\bfI & \bfzero \\
    \bfzero & -\bfI
    \end{pmatrix}
    \begin{pmatrix}
    \bfzero & -\bfI \\
    \bfI & \bfzero
    \end{pmatrix} \\
    &= &&
    \begin{pmatrix}
    \Real(\bfC+\bfE) & -\Imag(\bfC+\bfE) \\
    \Imag(\bfC+\bfE) & \Real(\bfC+\bfE) \\
    \end{pmatrix}
    +
    \begin{pmatrix}
    \Real(\bfD) & \Imag(\bfD) \\
    \Imag(\bfD) & -\Real(\bfD) \\
    \end{pmatrix} \\
    &= && \nabla^2 f(\bfx),
\end{aligned}
\end{equation}
where the second equality follows from Lemma~\ref{thm_block_mat_inv}.

In the non-trivial case, $\phi \neq (2k+1)\pi/2$ for any $k \in \intg$. Let
\begin{equation}
    \bfP \isdef
    \begin{pmatrix}
    \sqrt{\frac{1+\cos(2\phi)}{2}}\bfI & -\frac{\sin(2\phi)}{\sqrt{2(1+\cos(2\phi))}}\bfI \\
    \frac{\sin(2\phi)}{\sqrt{2(1+\cos(2\phi))}}\bfI & \sqrt{\frac{1+\cos(2\phi)}{2}}\bfI\\
    \end{pmatrix}.
\end{equation}
Likewise, the same equality holds.

Finally, we use the similarity property to show that an eigenvalue of $\nabla^2 f(\bfx)$ is also an eigenvalue of $\nabla^2f \left( \bfx\exp(j\phi) \right)$.

Let $(\lambda,\bfv)$ be an eigenpair of $\nabla^2 f(\bfx)$. The similarity property implies
\begin{equation}
\begin{aligned}
    & && \lambda\bfv = \nabla^2 f(\bfx) \bfv = \bfP^{-1} \nabla^2 f\left(\bfx\exp(j\phi)\right) \bfP \bfv \\
    &\implies && \nabla^2 f\left(\bfx\exp(j\phi)\right) \bfP \bfv = \lambda \bfP \bfv,
\end{aligned}
\end{equation}
i.e., $(\lambda,\bfP\bfv)$ is an eigenpair of $\nabla^2f \left( \bfx\exp(j\phi) \right)$. The proof in the other direction is straightforward.
\end{proof}

\section{Proof of Theorem~\ref{thm_general_case}} \label{app_thm_general_case}

Before we delve into the proof, it is useful to define a few auxiliary variables. Let
\begin{equation}
\begin{aligned}
    \tilde\bfC &\isdef
    \begin{pmatrix}
    \Real(\bfC) & -\Imag(\bfC) \\
    \Imag(\bfC) & \Real(\bfC)
    \end{pmatrix}, \\
    \tilde\bfD &\isdef
    \begin{pmatrix}
    \Real(\bfD) & \Imag(\bfD) \\
    \Imag(\bfD) & -\Real(\bfD)
    \end{pmatrix}, \\
    \tilde\bfE &\isdef
    \begin{pmatrix}
    \Real(\bfE) & -\Imag(\bfE) \\
    \Imag(\bfE) & \Real(\bfE)
    \end{pmatrix},
\end{aligned}
\end{equation}
so that $\nabla^2f(\bfx) = \tilde\bfC + \tilde\bfD + \tilde\bfE$. Likewise, $\tilde\bfC, \tilde\bfD, \tilde\bfE : \cplx^n \to \real^{2n \times 2n}$ are de facto (composite) functions of $\bfx$. The proof of the main theorem is based on the following minor result.
\begin{lemma} \label{thm_mat_prod}
For any $\bfx \isdef \bfx_R + j\bfx_I$, the following equalities hold:
\begin{equation}
\begin{aligned}
    \begin{pmatrix}
    \bfx_R^T & \bfx_I^T
    \end{pmatrix}
    \tilde\bfC
    \begin{pmatrix}
    \bfx_R \\
    \bfx_I
    \end{pmatrix}
    &= \bfx^H\bfC\bfx, \\
    \begin{pmatrix}
    \bfx_R^T & \bfx_I^T
    \end{pmatrix}
    \tilde\bfD
    \begin{pmatrix}
    \bfx_R \\
    \bfx_I
    \end{pmatrix}
    &= \Real\left( \bfx^H\bfD\conj{\bfx} \right), \\
    \begin{pmatrix}
    \bfx_R^T & \bfx_I^T
    \end{pmatrix}
    \tilde\bfE
    \begin{pmatrix}
    \bfx_R \\
    \bfx_I
    \end{pmatrix}
    &= \bfx^H\bfE\bfx, \\
    \tilde\bfC
    \begin{pmatrix}
    \bfx_R \\
    \bfx_I
    \end{pmatrix}
    &=
    \begin{pmatrix}
    \Real(\bfC\bfx) \\
    \Imag(\bfC\bfx)
    \end{pmatrix}, \\
    \tilde\bfD
    \begin{pmatrix}
    \bfx_R \\
    \bfx_I
    \end{pmatrix}
    &=
    \begin{pmatrix}
    \Real(\bfD\conj{\bfx}) \\
    \Imag(\bfD\conj{\bfx})
    \end{pmatrix}, \\
    \tilde\bfE
    \begin{pmatrix}
    \bfx_R \\
    \bfx_I
    \end{pmatrix}
    &=
    \begin{pmatrix}
    \Real(\bfE\bfx) \\
    \Imag(\bfE\bfx)
    \end{pmatrix}.
\end{aligned}
\end{equation}
\end{lemma}
\begin{proof}
The proof follows via straightforward computations.
\end{proof}

\begin{proof}[Proof of Theorem~\ref{thm_general_case}]
\eqref{thm_general_case_zero}
Since $\nabla f(\bfzero)=\bfzero$, $\bfzero$ is a critical point. Observe that
\begin{equation}
\begin{aligned}
    \bfC(\bfzero) &= \bfD(\bfzero) = \bfzero, \\
    \bfE(\bfzero) &= -\bfA^H\Diag(\bfb)\bfB - \bfB^H\Diag(\conj{\bfb})\bfA.
\end{aligned}
\end{equation}
For any $\bfx \isdef \bfx_R + j\bfx_I \neq \bfzero$,
\begin{equation}
\begin{aligned}
    & && \begin{pmatrix}
    \bfx_R^T & \bfx_I^T
    \end{pmatrix}
    \nabla^2 f(\bfzero)
    \begin{pmatrix}
    \bfx_R \\
    \bfx_I
    \end{pmatrix} \\
    &= &&
    \begin{pmatrix}
    \bfx_R^T & \bfx_I^T
    \end{pmatrix}
    \tilde\bfE(\bfzero)
    \begin{pmatrix}
    \bfx_R \\
    \bfx_I
    \end{pmatrix} \\
    &= &&
    \bfx^H \bfE(\bfzero) \bfx,
\end{aligned}
\end{equation}
where the second equality is given by Lemma~\ref{thm_mat_prod}. Since $\bfE(\bfzero)$ is Hermitian, we have
\begin{equation}
\begin{cases}
\bfx^H \bfE(\bfzero) \bfx > 0 & \text{if } \bfA^H\Diag(\bfb)\bfB + \bfB^H\Diag(\conj{\bfb})\bfA \prec \bfzero \\
\bfx^H \bfE(\bfzero) \bfx < 0 & \text{if } \bfA^H\Diag(\bfb)\bfB + \bfB^H\Diag(\conj{\bfb})\bfA \succ \bfzero
\end{cases}
\end{equation}
for any $\bfx\in\cplx^n\setminus\{\bfzero\}$.

\eqref{thm_general_case_cond}
Suppose $\exists\bfz\neq\bfzero$ such that $\nabla f(\bfz) = \bfzero$. \eqref{eq_nls_kkt} implies
\begin{equation} \label{thm_general_case_cond_grad}
\begin{aligned}
    \bfd(\bfz) &= {}&& \bfA^H\left( ( (\bfA\bfz) \had (\conj{\bfB\bfz}) - \bfb ) \had (\bfB\bfz) \right) +{} \\
    & && \bfB^H\left( ( (\conj{\bfA\bfz}) \had (\bfB\bfz) - \conj{\bfb} ) \had (\bfA\bfz) \right) \\
    &= && \bfzero.
\end{aligned}
\end{equation}
A few manipulations lead to
\begin{equation}
\begin{aligned}
    & \bfA^H\left( (\bfA\bfz) \had (\conj{\bfB\bfz}) \had (\bfB\bfz) \right) + \bfB^H\left( (\conj{\bfA\bfz}) \had (\bfB\bfz) \had (\bfA\bfz) \right) \\
    ={}& \bfA^H \Diag(\bfb) \bfB\bfz + \bfB^H \Diag(\conj{\bfb}) \bfA\bfz.
\end{aligned}
\end{equation}
Multiplying both sides with $\bfz^H$ on the left yields
\begin{equation}
\begin{aligned}
     & \bfz^H \left( \bfA^H\Diag(\bfb)\bfB + \bfB^H\Diag(\conj{\bfb})\bfA \right) \bfz \\
     ={}& 2\|(\bfA\bfz) \had (\conj{\bfB\bfz})\|_2^2 \geq 0,
\end{aligned}
\end{equation}
which implies $\bfA^H\Diag(\bfb)\bfB + \bfB^H\Diag(\conj{\bfb})\bfA \nprec \bfzero$.

\eqref{thm_general_case_rank}
We prove that $\nabla^2f(\bfz)$ is rank deficient by showing $\nabla^2f(\bfz)
\begin{pmatrix}
\bfx_R \\
\bfx_I
\end{pmatrix} = \bfzero$, where $\bfx \isdef \bfz_I - j\bfz_R$. Observe
\begin{enumerate}
    \item $(\bfA\bfx) \had (\conj{\bfB\bfz}) + (\bfA\bfz) \had (\conj{\bfB\bfx}) = \bfzero$.
    \item For any $\bfF$, $\bfF\bfz=\bfzero \implies \bfF\bfx=\bfzero$, which together with \eqref{thm_general_case_cond_grad} implies
\end{enumerate}
\begin{equation}
\begin{aligned}
     \bfzero ={}& \bfA^H \Diag\left( (\bfA\bfz) \had (\conj{\bfB\bfz}) - \bfb \right) \bfB \bfx +{} \\
     & \bfB^H \Diag\left( (\conj{\bfA\bfz}) \had (\bfB\bfz) - \conj{\bfb} \right) \bfA \bfx.
\end{aligned}
\end{equation}
By Lemma~\ref{thm_mat_prod}, it suffices to check $\bfC\bfx+\bfD\conj{\bfx}+\bfE\bfx$:
\begin{equation}
\begin{aligned}
     & && \bfC\bfx+\bfD\conj{\bfx}+\bfE\bfx \\
     &= && \bfA^H \Diag\left( (\bfB\bfz) \had (\conj{\bfB\bfz}) \right) \bfA\bfx +{} \\
     & && \bfB^H \Diag\left( (\bfA\bfz) \had (\conj{\bfA\bfz}) \right) \bfB\bfx +{} \\
     & && \bfA^H \Diag\left( (\bfA\bfz) \had (\bfB\bfz) \right) \conj{\bfB}\conj{\bfx} +{} \\
     & && \bfB^H \Diag\left( (\bfA\bfz) \had (\bfB\bfz) \right) \conj{\bfA}\conj{\bfx} +{} \\
     & && \bfA^H \Diag\left( (\bfA\bfz) \had (\conj{\bfB\bfz}) - \bfb \right) \bfB\bfx +{} \\
     & && \bfB^H \Diag\left( (\conj{\bfA\bfz}) \had (\bfB\bfz) - \conj{\bfb} \right) \bfA\bfx \\
     &= && \bfA^H \Diag(\bfB\bfz) \left( (\bfA\bfx) \had (\conj{\bfB\bfz}) + (\bfA\bfz) \had (\conj{\bfB\bfx}) \right) +{} \\
     & && \bfB^H \Diag(\bfA\bfz) \left( (\conj{\bfA\bfz}) \had (\bfB\bfx) + (\conj{\bfA\bfx}) \had (\bfB\bfz) \right) \\
     &= && \bfzero.
\end{aligned}
\end{equation}

\eqref{thm_general_case_num}
For any $\phi\in\real$, it is obvious that
\begin{equation}
\begin{aligned}
    \bfd\left(\bfz\exp(j\phi)\right) &= \bfd(\bfz) = \bfzero, \\
    f\left(\bfz\exp(j\phi)\right) &= f(\bfz),
\end{aligned}
\end{equation}
i.e., $\bfz\exp(j\phi)$ is also a nonzero critical point that is as good. By Proposition~\ref{thm_eig_hess}, $\nabla^2f \left( \bfx\exp(j\phi) \right)$ has the same eigenvalues as $\nabla^2f(\bfx)$ and therefore the same definiteness.

\end{proof}

\section{Proof of Theorem~\ref{thm_special_case}} \label{app_thm_special_case}

Without loss of generality, assume that $\|\bfA\had\conj{\bfB}\| \neq 0$.
\begin{proof}[Proof of Thm.~\ref{thm_special_case}]
\eqref{thm_special_case_zero}
When $n=1$, observe that
\begin{equation}
    \bfA^H\Diag(\bfb)\bfB + \bfB^H\Diag(\conj{\bfb})\bfA = 2 \Real\left( (\bfA\had\conj{\bfB})^H b \right).
\end{equation}
The rest follows directly from Thm.~\ref{thm_general_case}\eqref{thm_general_case_zero}.

\eqref{thm_special_case_cond}
For any $\bfz\in\cplx\setminus\{0\}$, we have
\begin{equation}
\begin{aligned}
    \bfd(\bfz) &= {}&& \bfA^H\left( ( (\bfA\bfz) \had (\conj{\bfB\bfz}) - \bfb ) \had (\bfB\bfz) \right) +{} \\
    & && \bfB^H\left( ( (\conj{\bfA\bfz}) \had (\bfB\bfz) - \conj{\bfb} ) \had (\bfA\bfz) \right) \\
    &= && \bfz|\bfz|^2 \bfA^H(\bfA\had\conj{\bfB}\had\bfB) - \bfz\bfA^H(\bfb\had\bfB) +{} \\
    & && \bfz|\bfz|^2 \bfB^H(\conj{\bfA}\had\bfB\had\bfA) - \bfz\bfB^H(\conj{\bfb}\had\bfA) \\
    &= && \bfz|\bfz|^2 \|\bfA\had\conj{\bfB}\|_2^2 - \bfz (\bfA\had\conj{\bfB})^H \bfb +{} \\
    & && \bfz|\bfz|^2 \|\bfA\had\conj{\bfB}\|_2^2 - \bfz(\conj{\bfA}\had\bfB)^H \conj{\bfb} \\
    &= && 2\bfz \left( \|\bfA\had\conj{\bfB}\|_2^2 |\bfz|^2 - \Real\left( (\bfA\had\conj{\bfB})^H \bfb \right) \right),
\end{aligned}
\end{equation}
which has a nonzero root if and only if $\Real\left( (\bfA\had\conj{\bfB})^H \bfb \right)>0$. If this condition is satisfied, its power is given by
\begin{equation}
    |\bfz|^2 = \Real\left( (\bfA\had\conj{\bfB})^H \bfb \right) / \|\bfA\had\conj{\bfB}\|_2^2.
\end{equation}

\eqref{thm_special_case_rank}
Suppose $\exists\bfz\neq\bfzero$ such that $\nabla f(\bfz) = \bfzero$. \eqref{eq_nls_kkt} and \eqref{thm_special_case_cond} imply
\begin{equation} \label{eq_special_case_cond}
    \|\bfA\had\conj{\bfB}\|_2^2 |\bfz|^2 - \Real\left( (\bfA\had\conj{\bfB})^H \bfb \right) = \bfzero.
\end{equation}

By Lemma~\ref{thm_mat_prod}, we have for any $\bfx \isdef \bfx_R + j\bfx_I \neq 0$
\begin{equation}
\begin{aligned}
    & && \begin{pmatrix}
    \bfx_R^T & \bfx_I^T
    \end{pmatrix}
    \nabla^2 f(\bfz)
    \begin{pmatrix}
    \bfx_R \\
    \bfx_I
    \end{pmatrix} \\
    &= && \bfx^H\bfC(\bfz)\bfx + \Real\left(\bfx^H\bfD(\bfz)\conj{\bfx}\right) + \bfx^H\bfE(\bfz)\bfx \\
    &= && |\bfx|^2 \bfA^H \Diag\left( (\bfB\bfz) \had (\conj{\bfB\bfz}) \right) \bfA +{} \\
    & && |\bfx|^2 \bfB^H \Diag\left( (\bfA\bfz) \had (\conj{\bfA\bfz}) \right) \bfB +{} \\
    & && \Real\left( \conj{\bfx}^2 \bfA^H \Diag\left( (\bfA\bfz) \had (\bfB\bfz) \right) \conj{\bfB} \right) +{} \\
    & && \Real\left( \conj{\bfx}^2 \bfB^H \Diag\left( (\bfA\bfz) \had (\bfB\bfz) \right) \conj{\bfA} \right) +{} \\
    & && |\bfx|^2 \bfA^H \Diag\left( (\bfA\bfz) \had (\conj{\bfB\bfz}) - \bfb \right) \bfB +{} \\
    & && |\bfx|^2 \bfB^H \Diag\left( (\conj{\bfA\bfz}) \had (\bfB\bfz) - \conj{\bfb} \right) \bfA \\
    &= && 2 \|\bfA\had\conj{\bfB}\|_2^2 |\bfx|^2 |\bfz|^2 + 2 \|\bfA\had\conj{\bfB}\|_2^2 \Real(\conj{\bfx}^2 \bfz^2 ) +{} \\
    & && 2|\bfx|^2 \left( \|\bfA\had\conj{\bfB}\|_2^2 |\bfz|^2 - \Real\left( (\bfA\had\conj{\bfB})^H \bfb \right) \right) \\
    &= && \|\bfA\had\conj{\bfB}\|_2^2 \left( 2 |\bfx|^2 |\bfz|^2 + 2 \Real(\conj{\bfx}^2 \bfz^2 ) \right) \\
    &= && \|\bfA\had\conj{\bfB}\|_2^2 (\conj{\bfx}\bfz + \bfx\conj{\bfz})^2 \geq 0,
\end{aligned}
\end{equation}
where the equality can be attained with $\bfx = \bfz_I - j\bfz_R$.

That $\nabla^2f(\bfz)$ is rank deficient is implied by Thm.~\ref{thm_general_case}\eqref{thm_general_case_rank}.

\eqref{thm_special_case_num}
This follows directly from Thm.~\ref{thm_general_case}\eqref{thm_general_case_num}.

\eqref{thm_special_case_local}
Let us perturb $\bfz$ by $\epsilon\in\cplx$ with $|\epsilon|$ being arbitrarily small and observe
\begin{equation}
\begin{aligned}
    f(\bfz+\epsilon) &= && \frac{1}{2} \| \left(\bfA (\bfz+\epsilon)\right) \had \left(\conj{(\bfB \bfz+\epsilon)}\right) - \bfb \|_2^2 \\
    &= && \frac{1}{2} \| (\bfA\bfz) \had (\conj{\bfB\bfz}) + (\bfA\bfz) \had (\conj{\bfB\epsilon}) +{} \\
    & && \phantom{\frac{1}{2} \|} (\bfA\bfepsilon) \had (\conj{\bfB\bfz}) + (\bfA\bfepsilon) \had (\conj{\bfB\bfepsilon}) - \bfb \|_2^2 \\
    &= && \frac{1}{2} \| (\bfA\bfz) \had (\conj{\bfB\bfz}) - \bfb + \bfA \had \conj{\bfB} (\conj{\epsilon}\bfz + \epsilon \conj{\bfz} + |\epsilon|^2) \|_2^2 \\
    &= && \frac{1}{2} \| (\bfA\bfz) \had (\conj{\bfB\bfz}) - \bfb \|_2^2 +{} \\
    & && \frac{1}{2} \| \bfA \had \conj{\bfB} (2\Real(\conj{\epsilon}\bfz) + |\epsilon|^2) \|_2^2 +{} \\
    & && \Real\Big( \left((\bfA\bfz) \had (\conj{\bfB\bfz}) - \bfb\right)^H \cdot \\
    & && \phantom{\Real\Big(\,} \left(\bfA \had \conj{\bfB} (2\Real(\conj{\epsilon}\bfz) + |\epsilon|^2)\right) \Big) \\
    &= && \frac{1}{2} \| (\bfA\bfz) \had (\conj{\bfB\bfz}) - \bfb \|_2^2 +{} \\
    & && \frac{1}{2} \| \bfA \had \conj{\bfB} \|_2^2 (2\Real(\conj{\epsilon}\bfz) + |\epsilon|^2)^2 +{} \\
    & && \Real\left( \|\bfA\had\conj{\bfB}\|_2^2 |\bfz|^2 - \bfb^H (\bfA\had\conj{\bfB}) \right) \cdot \\
    & && (2\Real(\conj{\epsilon}\bfz) + |\epsilon|^2) \\
    &= && f(\bfz) + \frac{1}{2} \| \bfA \had \conj{\bfB} \|_2^2 (2\Real(\conj{\epsilon}\bfz) + |\epsilon|^2)^2 +{} \\
    & && \left( \|\bfA\had\conj{\bfB}\|_2^2 |\bfz|^2 - \Real\left( (\bfA\had\conj{\bfB})^H \bfb \right) \right) \cdot \\
    & && (2\Real(\conj{\epsilon}\bfz) + |\epsilon|^2) \\
    &= && f(\bfz) + \frac{1}{2} \| \bfA \had \conj{\bfB} \|_2^2 (2\Real(\conj{\epsilon}\bfz) + |\epsilon|^2)^2,
\end{aligned}
\end{equation}
where the last equality is given by \eqref{eq_special_case_cond}. As a result,
\begin{equation}
    f(\bfz+\epsilon) - f(\bfz) = \frac{1}{2} \| \bfA \had \conj{\bfB} \|_2^2 (2\Real(\conj{\epsilon}\bfz) + |\epsilon|^2)^2 \geq 0,
\end{equation}
i.e., $\bfz$ is a local minimum.
\end{proof}

\section*{Acknowledgment}

TerraSAR-X data was provided by the German Aerospace Center (DLR) under the TerraSAR-X New Modes AO Project LAN2188.

N.~Ge would like to thank Dr.~S.~Auer and Dr.~C.~Gisinger for preparing Level~0 ground truth data, Dr.~S.~Li for the insight from an algebraic geometric point of view, F.~Rodriguez~Gonzalez and A.~Parizzi for valuable discussions, Dr.~H.~Ansari and Dr.~S.~Montazeri for explaining the sequential estimator and EMI, Dr.\ F. De Zan for providing an overview of multi-look multi-master SAR tomography in forested scenarios, Dr.~Y.~Huang and Prof.~Dr.~L.~Ferro-Famil for discussing a preliminary version of this paper, and the reviewers for their constructive and insightful comments.
\ifCLASSOPTIONcaptionsoff
  \newpage
\fi



\bibliographystyle{IEEEtran}
\bibliography{ref}
\end{document}